\documentclass{article}

    \PassOptionsToPackage{numbers, compress}{natbib}

\usepackage[preprint]{neurips_2022}

\pdfoutput=1

\usepackage[utf8]{inputenc} 
\usepackage[T1]{fontenc}    
\usepackage{hyperref}       
\usepackage{url}            
\usepackage{booktabs}       
\usepackage{amsfonts}       
\usepackage{nicefrac}       
\usepackage{microtype}      
\usepackage{xcolor}         

\usepackage{amsmath}
\usepackage{amssymb}
\usepackage{mathtools}
\usepackage{amsthm}
\usepackage[linesnumbered,ruled,vlined]{algorithm2e}
\usepackage[page, header]{appendix}

\SetKwInput{KwInput}{Input}                
\SetKwInput{KwOutput}{Output}              

\SetCommentSty{mycommfont}
\usepackage[capitalize,noabbrev]{cleveref}

\usepackage{xcolor}
\usepackage{amsfonts}
\usepackage{graphicx}
\usepackage[T1]{fontenc}
\usepackage{wrapfig}
\usepackage{subfigure}
\usepackage{multirow}

\theoremstyle{plain}
\newtheorem{theorem}{Theorem}[section]
\newtheorem{proposition}[theorem]{Proposition}
\newtheorem{lemma}[theorem]{Lemma}
\newtheorem{corollary}[theorem]{Corollary}
\theoremstyle{definition}
\newtheorem{definition}[theorem]{Definition}

\theoremstyle{remark}

\newcommand{\framework}{EPSRO}

\title{Efficient Policy Space Response Oracles}

\author{
  Ming Zhou\thanks{Corresponding author: \href{mingak@sjtu.edu.cn}{mingak@sjtu.edu.cn}} \textsuperscript{ \rm 1,4} \\
  \And
  Jingxiao Chen\textsuperscript{ \rm 1}\\
  \And
  Ying Wen\textsuperscript{ \rm 1}\\
  \And
  Weinan Zhang\textsuperscript{ \rm 1}\\
  \AND
  Yaodong Yang\textsuperscript{ \rm 2}
  \And
  Yong Yu\textsuperscript{ \rm 1}\\
  \And
  Jun Wang\textsuperscript{ \rm 3,4}\\
  \AND
  \\
  \vspace{-25pt}
  \textsuperscript{\rm 3}University College London,\textsuperscript{\rm 4}Shanghai Digital Brain Laboratory\\
 \textsuperscript{\rm 1}Shanghai Jiao Tong University, \textsuperscript{\rm 2}Institute for Artificial Intelligence, Peking University\\
\setcounter{footnote}{0}
}

\begin{document}

\maketitle

\begin{abstract}
Policy Space Response Oracle methods (PSRO) provide a general solution to learn Nash equilibrium in two-player zero-sum games but suffer from two drawbacks: (1) the \textit{computation inefficiency} due to the need for consistent meta-game evaluation via simulations, and (2) the \textit{exploration inefficiency} due to finding the best response against a fixed meta-strategy at every epoch. In this work, we propose Efficient PSRO (\framework{}) that largely improves the efficiency of the above two steps. Central to our development is the newly-introduced subroutine of \textit{no-regret optimization} on the \textit{unrestricted-restricted (URR)} game. By solving URR at each epoch, one can evaluate the current game and compute the best response in one forward pass without the need for meta-game simulations. Theoretically, we prove that the solution procedures of  \framework{} offer a monotonic improvement on the exploitability, which none of existing PSRO methods possess. Furthermore, we prove that the no-regret optimization has a regret bound of $\mathcal{O}(\sqrt{T\log{[(k^2+k)/2]}})$, where $k$ is the size of restricted policy set. Most importantly, a desirable property of  \framework{} is that it is parallelizable, this allows for highly efficient exploration in the policy space that induces behavioral diversity. We test  \framework{} on three classes of games, and report a 50x speedup in wall-time and 10x data efficiency while maintaining similar exploitability as existing PSRO methods on Kuhn and Leduc Poker games.
\end{abstract}
\section{Introduction}
Policy Space Response Oracles (PSRO)~\cite{lanctot2017unified} is a general multi-agent reinforcement learning algorithm, which has been applied in many non-trivial multi-agent learning tasks~\cite{vinyals2019grandmaster,berner2019dota,liu2021towards}.
In general, PSRO aims to find an approximate Nash equilibrium (NE) by iteratively expanding a restricted game formed by a set of restricted policy sets, which is ideally much smaller than the original game.
At each epoch, PSRO executes sequential learning composed of a meta-game solving and a learning of best responses.
Though PSRO does not need to learn policies in the original game directly, the learning of PSRO is still inefficient in solving meta-game and learning high-quality best responses.

Specifically, PSRO is \textbf{computation inefficient} to solve a geometrically growing meta-game because it relies on numerous simulations across the Cartesian space of growing policy sets~\cite{omidshafiei2019alpha,yang2020alpha}.
Moreover, learning against a fixed opponent meta-strategy to find a best response is \textbf{exploration inefficient}. In such a way, PSRO has no non-degenerate guarantee on the expansion of restricted policy sets, since the fixed meta-strategy is only a best response to the restricted policy set~\cite{wang2021evaluating}.
Despite playing against fixed opponent meta-strategies can theoretically expand the policy sets~\cite{mcaleer2020pipeline}, it has no guarantee that the discovered best responses can still hold the strength when opponents deviate their strategies.
Therefore, PSRO needs to add all possible policies from the original game and generates large restricted policy sets in the worst case, making it slow to converge~\cite{mcaleer2020pipeline}.
A straightforward idea to improve the learning efficiency is to utilize parallelism for the learning of best responses~\cite{lanctot2017unified,balduzzi2019open,mcaleer2020pipeline,liu2021neupl}.
However, all of these existing methods still require simulations to solve meta-game and learn best responses by playing against fixed meta-strategies.

\begin{figure*}[ht!]
    \begin{center}
	\centerline{\includegraphics[width=1.0\columnwidth]{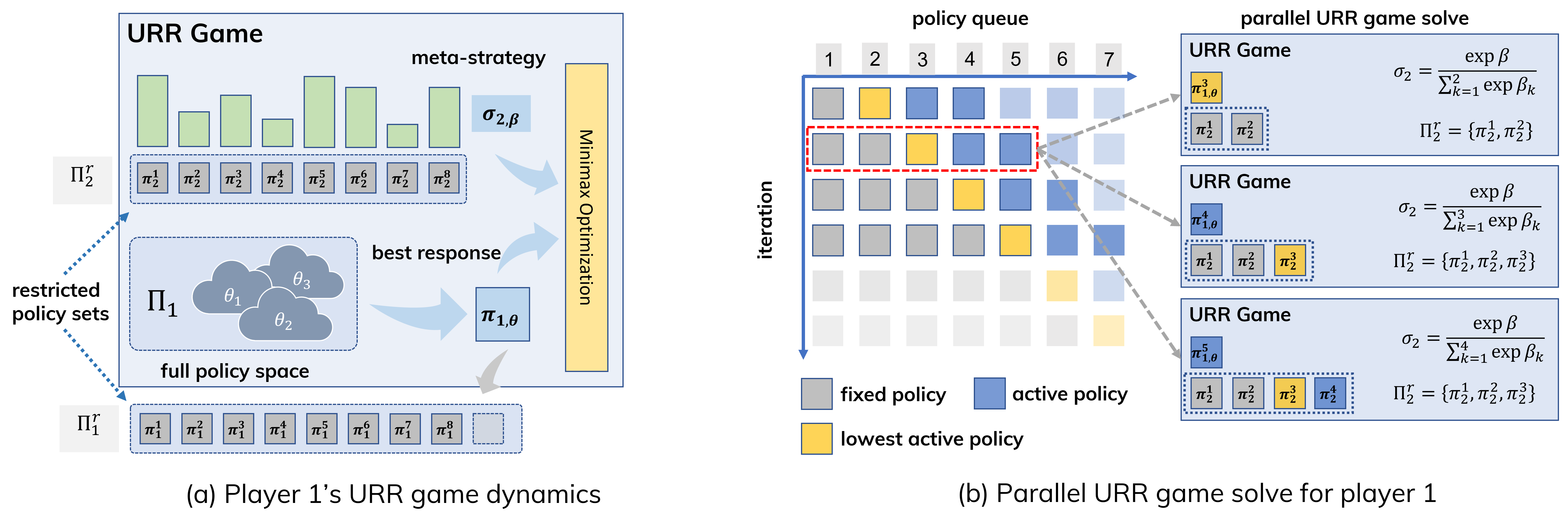}}
    \caption{An overview of \framework{}. (a) \framework{} runs in a loop that learns best responses by solving URR games for each player, then expands the restricted policy sets with these $\pi_{i,\theta}$ at each epoch. (b) At each epoch, \framework{} runs multiple URR game solves for each active best response $\pi^j_{i,\theta}$ in parallel (Algorithm~\ref{alg:pipeline}), where $j$ is the level. In each URR game, $\pi^j_{i,\theta}$ plays against $\Pi^{r,k}_{-i}$ where $k=1,\dots,j-1$.}
    \label{fig:pepsro}
    \end{center}
\end{figure*}

It is desirable for an efficient method to all of these problems.
Our key insight is that the computation of meta-strategies can be free from simulations, and the learning of best responses should be toward monotonic expansion on restricted policy sets. As for the learning of best responses, there have several proposals for building an \textit{unrestricted-restricted game} (URR game, Section~\ref{sec:URR}) to learn a generalized best response as a Nash equilibrium to the opponent's restricted policy 
set~\cite{zinkevich2007new,hansen2008range}. However, these existing works focus on tabular cases in which the policy space is limited.

In this paper, we introduce Efficient PSRO (\framework{}) based on the URR to save the simulation cost. 
\framework{} does not require simulations to compute meta-strategies beforehand. Instead, at each epoch, \framework{} solves URR games to learn a best response and an opponent meta-strategy, which is an approximate NE. 
In addition, we prove that the learned best responses are guaranteed to expand the restricted policy sets in a non-degenerate manner, improving the exploration efficiency.
As for the URR solve, it is built on top of no-regret optimization~\cite{daskalakis2011near}. Moreover, we propose an efficient warm-start technique for the regret optimization to save the re-training cost caused by the changed length of meta-strategies. We analyze the algorithm performance and give a regret bound of $\mathcal{O}(\sqrt{T\log{[(k^2+k)/2]}})$, where $k$ is the size of the final restricted policy set. Most importantly, we introduce a pipeline URR solver to make the best response learning parallelizable, which further improves the exploration efficiency. 
The demonstration shows that \framework{} substantially improves the training efficiency and achieves better performance than existing PSRO-based methods in high-dimensional matrix games, poker games, and multi-agent gathering tasks.


\section{Preliminaries}\label{sec:preliminaries}
\paragraph{Two-player Normal-form Games.} A two-player normal-form game~\cite{fudenberg1991game} is a tuple $(\Pi, U^{\Pi})$, where $\Pi=(\Pi_1, \Pi_2)$ and $U^{\Pi}=(U^{\Pi_1}, U^{\Pi_2})$ are the tuple of policy sets and the tuple of payoff tables, respectively.
Formally, $\forall i \in \{1, 2\}, U^{\Pi_i}: \Pi \rightarrow \mathbb{R}^{|\Pi_1|\times|\Pi_2|}$, in which each item represents the utility of a joint policy.
Players in the game try to maximize their own expected utility by sampling policy from a mixture (distribution) $\sigma_i$ over their policy sets, where $\forall i \in \{1, 2\}, \sigma_i \in \Delta(\Pi_i)$.
For the sake of convenience, we use $-i$ to denote the other agent except for player $i$ in the following content.
A best response to a mixed-strategy $\sigma_{-i}$ is defined as a strategy that has highest utility.
It can be expressed as $\textbf{BR}(\sigma_{-i}) = \arg\max_{\sigma'_i}u_i(\sigma'_i, \sigma_{-i})$, where $u_i(\cdot,\cdot)$ represents the utility function of player $i$ for a given joint policy.

\paragraph{Policy Space Response Oracles (PSRO)}
Double Oracle (DO) methods~\cite{mcmahan2003planning,le2021online,mcaleer2021xdo} provide an iterative mechanism for finding a Nash equilibrium approximation in normal form games. These algorithms work in expanding a restricted policy set $\Pi^r_i$ for each player iteratively.
At each epoch, a Nash equilibrium $\sigma=(\sigma_i,\sigma_{-i})$ is computed for a restricted game which is formed by a tuple of restricted policy sets $\Pi^r=(\Pi^r_i, \Pi^r_{-i})$.
Then, a best response to this Nash equilibrium for each player $i$ is computed and added to its restricted policy set $\Pi^r_i = \Pi^r_i \cup \{\textbf{BR}(\sigma_{-i})\}$.
PSRO is a generalization of DO where the restricted game's choice is a policy rather than an action.
At each epoch, PSRO learns an approximate best response to an Nash equilibrium via the oracles (e.g., reinforcement learning algorithms).
There are many different solvers for the computation of Nash equilibrium, such as $\alpha$-rank~\cite{omidshafiei2019alpha}, PRD~\cite{lanctot2017unified} or some linear programming methods~\cite{sandholm2005mixed}.
In practice, PSRO seeks an approximation of Nash equilibrium, which is at a level of precision $\epsilon \ge 0$~\cite{aziz2010multiagent}.
To evaluate the equality of approximation, we use $\textsc{NashConv}(\sigma)=\sum_{i}u_i\left(\textbf{BR}_i(\sigma_{-i}), \sigma_{-i}\right) - u_i(\sigma)$ to compute the exploitability of $\sigma$ to an oracle $\{\textbf{BR}(\sigma_{-i})\}$~\cite{johanson2011accelerating}.
$\sigma$ is an exact Nash equilibrium if $\textsc{NashConv}=0$.

We summarize the pseudo code of PSRO in Algorithm~\ref{alg:psro}.
At each epoch, PSRO requires simulations to compute the missing items in $U^{\Pi^r}$ after the learning of best responses, which causes an expensive computing cost.
In general, the amount of simulations grows geometrically as $\mathcal{O}(M\cdot|\Pi^r_i|)$, where $|\Pi^r_i|$ and $M$ denote the size of restricted policy set and the number of simulations for each missing item, respectively. To learn approximate best responses, PSRO usually runs nested reinforcement learning algorithms. However, such a procedure is data-thirsty and has no guarantee to find a high-quality best response to bring higher payoffs for a restricted policy set, especially in the case of complex tasks.

\begin{minipage}[b]{0.60\linewidth}
\begin{algorithm}[H]
 \caption{\textsc{Vanilla PSRO}}\label{alg:psro}
 \KwInput{initial restricted policy sets $\Pi^r=(\Pi^r_{1}, \Pi^r_{2})$}
 \tcc{can be saved via URR games}
 \KwInput{\textcolor{orange}{empty payoff table $U^{\Pi^r}$}}
 \KwInput{meta-strategies $\sigma_i \sim \textsc{Uniform}(\Pi^r_i)$}
 	\While{not terminated}{
 		\For{player $i \in \{1, 2\}$}{
 		    \For{many episodes}{
 		        Train best response $\pi_{i,\theta}$ against $\pi_{-i} \sim \sigma_{-i}$
 		    }
 		    $\Pi^r_i = \Pi^r_i \cup \{\pi_i,\theta\}$
 		}
 		\tcc{can be saved via URR games}
 		\textcolor{orange}{Run simulations to compute missing entries in $U^{\Pi^r}$}
 
 		Compute a meta-strategy $\sigma$ from $U^{\Pi^r}$
 	}
 \KwOutput{current meta-strategy $\sigma_i$ for player $i$}
 \end{algorithm}
\end{minipage}
\begin{minipage}[b]{0.36\linewidth}
\begin{algorithm}[H]
 \caption{\textsc{Simplified PSRO with URR Games}}\label{alg:epsro}
 \KwInput{initial restricted policy sets $\Pi^r=(\Pi^r_{1}, \Pi^r_{2})$}
 	\While{not terminated}{
 		\For{player $i \in \{1, 2\}$}{
 			Random initialize a best response $\pi_{i, \theta}$
 			
 			\textcolor{orange}{($\pi_{i,\theta}$, $\sigma_{-i,\beta})$ = \textsc{SolveURR}($\pi_{i,\theta}$, $\Pi^r_{-i}$)}
 		}
 		$\Pi^r_i=\Pi^r_i \cup \{\pi_{i,\theta}\}$ for $i \in \{1,2\}$
 	}
 \KwOutput{current meta-strategy $\sigma_i$ for player $i$}
 \end{algorithm}
\end{minipage}

\section{EPSRO: Efficient PSRO}\label{sec:methodologies}
For \framework{}, the keys to improve its efficiency include two aspects: (1) eliminating simulations for computing meta-strategies to improve the computing efficiency; (2) finding high-quality best responses to improve the exploration efficiency.
We developed \framework{} on top of URR games (Section~\ref{sec:URR}) for the first aspect.
As for the exploration efficiency, it indicates the efficiency of expanding restricted policy sets.
In general, higher exploration efficiency means the algorithm can express a restricted policy space with a smaller policy set than other methods. Thus, the quality of learned best responses is vital to the exploration.
To tackle this problem, we propose an efficient algorithm to solve URR games in parallelism (Section~\ref{sec:solve_urr} and~\ref{sec:pepsro}).
We summarize the pseudo code of Efficient PSRO (\framework{}) in Algorithm~\ref{alg:pipeline} and give its overview in Figure~\ref{fig:pepsro}.

\subsection{Modeling \framework{} as URR Games}\label{sec:URR}
Saving the computational cost of simulation is crucial to optimizing the efficiency of PSRO methods. In \framework{}, we achieve that by developing a simulation-free mechanism for meta-strategy learning and policy space expansion. Specifically, the meta-strategies are derived from direct interaction with best responses $\textbf{BR}(\sigma_{-i}) \in \Delta(\Pi_i)$ instead of a simulation-based $U^{\Pi^r}$, and restricted policy sets are expanded with these \textbf{BR}s. We further model the interaction as an \textit{unrestricted-restricted game} below, which can be regarded as a parameterized extension of the tabular case in ~\cite{zinkevich2007new}.

\begin{definition}\label{def:urr}
	An \textit{unrestricted-restricted} (URR) game for player $i$ is a tuple of full policy set $\Pi_i$ and restricted policy set $\Pi^r_{-i}$, i.e. $(\Pi_i, \Pi^r_{-i})$. In this game, the player $i$ models its policy as a function parameterized by $\theta$, i.e. $\pi_{i,\theta} \in \Delta(\Pi_i)$. For each interaction, it plays against an opponent's policy $\pi_{-i}$ sampled from $\sigma_{-i,\beta} \in \Delta(\Pi^r_{-i})$, where $\sigma_{-i,\beta}$ is a meta-strategy parameterized by $\beta$. $(\pi^{\star}_{i,\theta}, \sigma^{\star}_{-i,\beta})$ is a Nash equilibrium if
	\begin{equation*}
		\pi^{\star}_{i,\theta} = \textbf{BR}(\sigma^{\star}_{-i,\beta})\text{, and } \sigma^{\star}_{-i,\beta} = \textbf{BR}(\pi^{\star}_{i,\theta}).
	\end{equation*}
\end{definition}

As described in Definition~\ref{def:urr}, the learning of best responses doesn't require fixed meta-strategies from a restricted game. Therefore, URR games save the computational cost of running simulations to construct the $U^{\Pi^r}$. Algorithm~\ref{alg:epsro} lists the pseudo-code of URR-based PSRO for the comparison with the vanilla PSRO (Algorithm ~\ref{alg:psro}).

Though URR can save the simulation costs, we need to ensure its policy set expansion is non-degenerate since the final $\Pi^r$ should be an approximation of complete policy space $\Pi$. To evaluate how closely the restricted policy set is to $\Pi$, Balduzzi et al.~\cite{balduzzi2019open} introduce the \textbf{gamescape}. Though the original cases are self-play, we can naturally bring this concept to URR games.

\begin{definition}[\emph{URR Gamescape, derived from~\cite{balduzzi2019open}}]
	Given an URR game $(\Pi_i,\Pi^r_{-i})$ (row player $i$, column player $-i$) with payoff matrix $\mathbf{U}$, the corresponding empirical gamescape (EGS) is $\mathcal{G} := \{\textit{convex mixture of columns of } \mathbf{U}\}$.
\end{definition}

Conceptually, the non-degenerate or monotonic policy set expansion means that the learned best responses should not be in existing $\mathcal{G}$, so that the algorithm can expand the restricted policy set to approach the complete policy space. We investigate the policy set expansion of URR games in the following theorem.

\begin{theorem}[\emph{Monotonic Policy Space Expanding}]\label{theorem:weak_strength}
	For any given epoch $e$ and $e+1$, let $(\pi^e_i, \sigma^e_{-i})$ and $(\pi^{e+1}_i, \sigma^{e+1}_{-i})$ be Nash equilibrium of $\textbf{URR}^e_i$ and $\textbf{URR}^{e+1}_i$, respectively, where $\pi^e_i,\pi^{e+1}_i \in \Pi_i$, $\sigma^e_{-i} \in \Delta^e_{\Pi^{r}_{-i}}$ and $\sigma^{e+1}_{-i} \in \Delta^{e+1}_{\Pi^{r}_{-i}}$.
	The utilities of $\pi^e_i$ against opponent strategies $\sigma^e_{-i}$ and $\sigma^{e+1}_{-i}$ satisfies
	\begin{equation}
		u_i(\pi^e_i, \mathbf{\sigma^e_{-i}}) - u_i(\pi^e_i, \mathbf{\sigma^{e+1}_{-i}}) \ge 0,
	\end{equation}
	where $\Delta^e_{\Pi^r_{-i}}$ indicates $\Delta(\Pi^{r,e}_{-i})$. Especially, $u_i(\pi^e_i,\sigma^e_{-i}) - u_i(\pi^e_i,\sigma^{e+1}_{-i}) > 0$ indicates there is a strictly policy space expanding at $e+1$, i.e., $\pi^{e+1}_{-i} \in \Pi^{r,e+1}_{-i}\setminus\Pi^{r,e}_{-i}$. (See Appendix~\ref{theorem:weak_strength_b})
\end{theorem}

Theorem~\ref{theorem:weak_strength} shows that \framework{} can achieve a monotonic performance improvement and policy space expansion with URR games. We further investigate that \framework{} has higher exploration efficiency than the naive PSRO (Section~\ref{sec:exp}) and give a Proposition as follows.

\begin{proposition}\label{prop:higher_exploration}
EPSRO has higher exploration efficiency than PSRO. (See Appendix~\ref{prop:higher_exploration_b})
\end{proposition}

\subsection{Solving URR Games}\label{sec:solve_urr}
We've built our \framework{} with URR games and gave analysis for its policy set expansion, but how to solve the URR games is still a question. As solving Nash equilibrium in large-scale games is difficult, so we seek a learning-based method. We propose a multi-agent learning algorithm with the corresponding pseudo-code listed in Algorithm~\ref{alg:solve_urr}, which characterizes a procedure to train a best response and opponent meta-strategy for each player. Specifically, we learn the best response with a reinforcement learning algorithm while learning the meta-strategy with an online no-regret method.

 \begin{algorithm}[H]
 \caption{\textsc{SolveURR}}\label{alg:solve_urr}
\KwInput{URR game $(\Pi_i, \Pi^r_{-i})$, BR $\pi_{i,\theta} \sim \Delta(\Pi_i)$}
\KwInput{meta-strategies $\sigma_{-i,\beta} \sim \textsc{Uniform}(\Pi^r_{-i})$}
 	\While{not terminated}{
 		Train best response $\pi_{i,\theta'} \leftarrow \pi_{i,\theta} $ against $\pi_{-i} \sim \sigma_{-i,\beta}$ with reinforcement learning step
 		
 		Update $\sigma_{-i,\beta'} \leftarrow \sigma_{-i,\beta}$ against $\pi_{i,\theta}$ by following Algorithm~\ref{alg:deterministic} in Appendix~\ref{appendix:algorithm}
 		
 		$\pi_{i,\theta}\leftarrow \pi_{i,\theta'}$, $\sigma_{-i,\beta} \leftarrow \sigma_{-i, \beta'}$
 	}
\KwOutput{an approximate Nash $(\pi^{\star}_{i,\theta}, \sigma^{\star}_{-i,\beta})$}
 \end{algorithm}

In Algorithm~\ref{alg:solve_urr}, the opponent meta-strategy $\sigma_{-i}$ is represented as a Boltzman distribution, which is parameterized by a vector $\beta=\left[\beta_1, \beta_2, \dots, \beta_{|\Pi^r_{-i}|}\right]$. Thus, each support of $\sigma_{-i}$ could be expressed as $\sigma_{-i}(j)=\exp{\beta_j}/\sum^n_{i=1}\exp{\beta_i}$. We update $\sigma_{-i}$ by following Algorithm~\ref{alg:deterministic} in Appendix~\ref{appendix:parameters}. As for the best response $\pi_{i}$, it is a neural network trained with off-policy reinforcement learning~\cite{sutton2018reinforcement}.

A feasible method to quantify the learning performance of Algorithm~\ref{alg:solve_urr} is to calculate the regret to the oracle payoff.
In this paper, we use no-regret algorithms~\cite{bowling2004convergence,daskalakis2011near} to analyze \framework{}'s algorithm performance in two-player zero-sum games.
Under this framework, a learning algorithm could approximate the NE asymptotically by playing the same game repeatedly.
A well-known no-regret algorithm among them is Multiplicative Weights Update (MWU)~\cite{freund1999adaptive}, which updates strategy by considering the averaging loss along the learning horizon and then achieves no-regret as the learning horizon towards infinite. Though the update of \framework{}'s meta-strategies does not exploit the loss directly like MWU, it follows MWU by exploiting the computed gradients. We further explain it in Lemma~\ref{lemma:meta_update_b}.

\begin{definition}
	Considering a sequence of mixed strategies for player $i$ as $\pi_1,\pi_2,\dots$, an algorithm of $-i$ that generates a sequence of mixed strategies $\sigma_{1},\sigma_{2},\dots$ is called a no-regret algorithm if we have: $\lim_{T\rightarrow\infty}R_T/T=0,\text{ }R_T=\max_{\sigma\in\Delta_{\Pi_{-i}}}\sum^T_{t=1}(\pi^{\top}_tU\sigma-\pi^T_tU\sigma_t)$.
\end{definition}

\paragraph{Warm-Start Learning.} In \framework{}, a critical problem for the meta-strategy optimization is that the length of meta-strategy changes as the policy set expands. Therefore, we need to learn the meta-strategy with a new $\beta$ at each epoch. However, if we start the learning from scratch, the cost of re-training grows more and more expensive as the policy set expands, resulting in a slow convergence rate. To tackle this issue, we propose a warm-start technique that enables the meta-strategy starts from a non-trivial initialization.
Since the warm-start aims to save time to achieve a minimal regret when $\Pi^r$ changes, the key is to correctly initialize the regrets instead of only starting the strategy. Moreover, a wrong initialization of the regrets will result in huge regrets in subsequent iterations, which is no better than starting from scratch. As pointed in~\cite{brown2016strategy}, a feasible warm-start should be a substitute strategy that does not violate the regret bound and the approximate NE of the last epoch.

\begin{theorem}[\emph{Theorem 1 in~\cite{brown2016strategy}}]\label{theorem:averaging}
In a two-player zero-sum game, if $\frac{R^T_i}{T} \le \epsilon_i$ for both player $i\in\{1,2\}$, then $(\bar{\pi}_i, \bar{\sigma}_{-i})$ is a $(\epsilon_1 + \epsilon_2)$-equilibrium, where $\bar{\sigma}_{-i} = \sum^T_{t=1}\langle\sigma_{-i,t}, l_{-i,t}\rangle / \sum^T_{t=1}Tl_{-i,t}$, $\bar{\pi}_i = \sum^T_{t=1}\langle\pi_{i,t}, l_{i,t}\rangle / \sum^T_{t=1}Tl_{i,t}$, $R^T_i$ the summation of regrets of $T$ iterations, $l_t$ the loss vector.
\end{theorem}

Theorem~\ref{theorem:averaging} shows that if we use a regret-based average $\bar{\sigma}_{-i}$ (or $\bar{\pi}_i$) as the substitute of a sequence of $\sigma_{-i,t}$ (or $\pi_{i,t}$), the substitute can still hold the equilibrium. In case of \framework{}, as the restricted policy set grows from $\Pi^{r,e}_{-i}$ to $\Pi^{r,e+1}_{-i}$, we need to investigate whether there is a feasible substitute $\bar{\sigma}^{'}_{-i} \in \Delta^{e+1}_{\Pi^r_{-i}}$ to $\bar{\sigma}_{-i} \in \Delta^e_{\Pi^r_{-i}}$ follows the regret guarantee of epoch $e$.

\begin{theorem}[\emph{See Appendix~\ref{theorem:warm_bound_b}}]\label{theorem:warm_bound}
     Suppose a substitute policy of $\bar{\sigma}_{-i} \in \Delta^e_{\Pi^r_{-i}}$ is $\bar{\sigma}^{'}_{-i} \in \Delta^{e+1}_{\Pi^r_{-i}}$, and it satisfies $u^{e+1}_i(\bar{\pi}_i, \bar{\sigma}^{'}_{-i})=u^e_i(\bar{\pi}_i, \bar{\sigma}_{-i})$, we have
     \begin{equation*}
     	\max_{\sigma_{-i}\in\Delta^{e+1}_{\Pi^r_{-i}}}\sum^T_{t=1}\left( u^{e+1}_{-i}(\pi^t_i, \sigma_{-i}) - u^{e+1}_{-i}(\pi^t_i, \bar{\sigma}^{'}_{-i}) \right) \le \epsilon_{-i},
     \end{equation*}
     where $\epsilon_{-i}$ is the regret bound of epoch $e+1$, $\pi^t_i \in \Delta(\Pi_i)$.
\end{theorem}

Theorem~\ref{theorem:warm_bound} shows that there exists a substitute $\bar{\sigma}^{'}_{-i} \in \Delta^{e+1}_{\Pi^r_{-i}}$ holds Theorem~\ref{theorem:averaging} also holds the regret bound of epoch $e+1$. Thus, once we compute such a $\bar{\sigma}^{'}_{-i}$ that satisfies $u_i(\bar{\pi}_i, \bar{\sigma}^{'}_{-i})=u_i(\bar{\pi}_i, \bar{\sigma}_{-i})$, then it is a feasible warm-start at the next epoch to save the training time.

\begin{lemma}[\emph{See Appendix~\ref{lemma:beta_opt_b}}]\label{lemma:beta_opt}
	Let $k=|\Pi^{r,e+1}_{-i}|$, $\bar{\sigma}^{'}_{-i}$ be parameterized by $\beta_{-i}=\left[\beta_{-i,1},\beta_{-i,2},\dots,\beta_{-i,k}\right]$, $\bar{\sigma}^{'}_{-i}(k)$ the $k$-th item of $\bar{\sigma}^{'}_{-i}$, $x=[\bar{\sigma}^{'}_{-i}(1),\dots,\bar{\sigma}^{'}_{-i}(k - 1)]^T$, $\bar{l}^e_{-i}$ is $-i$'s average loss vector to $\bar{\pi}_i$ at epoch $e$. Then a feasible initial of $\beta^{e+1}_{-i}$ could satisfy
	\begin{equation}
		\beta^{e+1}_{-i} = \arg\min_{\beta_{-i}}\parallel (x-\bar{\sigma}_{-i})^{\top} \bar{l}^e_{-i} - \bar{\sigma}^{'}_{-i}(k) u_{-i}(\bar{\pi}_i,\pi^{e+1}_{-i})\parallel_2.
	\end{equation}
\end{lemma}

Lemma~\ref{lemma:beta_opt} computes a $\beta^{e+1}_{-i}$ for the initialization of $\sigma^{e+1}_{-i}$ to satisfy Theorem~\ref{theorem:averaging} and~\ref{theorem:warm_bound}. In addition, the best response $\pi_{i,\theta}$ also follows these conditions by continuous training instead of restarting parameters at each epoch. We now present the regret bound of \framework{} as follows.

\begin{theorem}[\emph{Regret Bound of \framework{}}]\label{theorem:regret_meta_warm}
Let $l_1,l_2,\dots,l_T$ be a sequence of loss vectors player by an opponent, and $\langle \cdot,\cdot \rangle$ be the dot product, then \framework{} is a no-regret algorithm with (See Appendix~\ref{theorem:regret_meta_warm_b})
	\begin{equation*}
		\frac{1}{T}\left(\sum^T_{t=1}\langle\sigma_t,l_t\rangle-\min_{\sigma \in \Delta(\Pi^r_t)}\sum^T_{t=1}\langle\sigma,l_t\rangle\right)\le \frac{\sqrt{\log{[(k+1)k/2]}}}{\sqrt{2T}},\text{ where }k\text{ is the size of }\Pi^r.
	\end{equation*}
\end{theorem}

Theorem~\ref{theorem:regret_meta_warm} shows that \framework{} is no-regret when the policy set is finite. Though some complex games have continuous policy space, the number of effective policies is finite. We further give the convergence rate of \framework{} as follows.

\begin{theorem}[\emph{Convergence Rate of \framework{}}]\label{theorem:convergence_rate_warm}
     Let $k$, $N$ denote the size of restricted policy sets $\Pi^r_{-i}$ and $\Pi_i$. Then the learning of Algorithm~\ref{alg:solve_urr} will converge to the Nash equilibrium with the rate (See Appendix~\ref{theorem:convergence_rate_warm_b}):
     \begin{equation*}
         \epsilon_T=\sqrt{\frac{\log{[(k+1)k/2]}}{2T}} + \sqrt{\frac{\log{[(N+1)N/2]}}{2T}}.
     \end{equation*}
\end{theorem}

\subsection{Pipeline URR Solver}\label{sec:pepsro}

\begin{wrapfigure}[16]{r}{0.6\linewidth}
	\vspace{-1em}
	\begin{algorithm}[H]
		\caption{\textsc{Efficient PSRO (EPSRO)}}\label{alg:pipeline}
		\KwInput{inital restricted policy sets $\Pi^r=(\Pi^r_1, \Pi^r_2)$}
		\While{not terminated}{
			\For{all player $i \in \{1, 2\}$ in parallel}{
				\For{loop all active best response $\pi^j_i \in \Pi^r_i$}{
					\For{all $\Pi^{r,<j}_{-i}$ in parallel}{
						$\pi^j_{i,\theta},\sigma^{<j}_{-i,\beta}=\textsc{SolveURR}(\pi^j_{i,\theta},\Pi^{r,<j}_{-i})$\hspace{-20pt}
						
						\If{the lowest $\pi^j_{i,\theta}$ meets stops cond.}{
							Set it to fixed and $\Pi^r_i := \Pi^r_i \cup \{\pi^j_i\}$
							
							Generate a new active policy
						}
					}
				}
			}
		}
		\KwOutput{current meta-strategies $\sigma=(\sigma_1, \sigma_2)$}
	\end{algorithm}
\end{wrapfigure}
While Algorithm~\ref{alg:solve_urr} is clear, the reinforcement learning step can take a long time to converge to a good response, especially in complex games. 
To overcome this training problem, we introduce a parallel solution to further improve \framework{}'s training efficiency.
Inspired by Pipeline-PSRO (P-PSRO)~\cite{mcaleer2020pipeline}, we schedule multiple asynchronous training procedures to train a best response policy for each player.
Figure~\ref{fig:pepsro} illustrates an example dynamics, it able to scale up the \emph{minimax optimization} with convergence guarantee by maintaining a hierarchical pipeline of reinforcement learning policies as P-PSRO (Proposition 3.2 in ~\cite{mcaleer2020pipeline}). We finally give \framework{}'s pseudo-code with parallel URR solve in Algorithm~\ref{alg:pipeline}. Each player $i$ in parallel maintains a queue of ordered policies with two classes of training policies: fixed policies at low levels and active policies at high levels. Each active policy $\pi^j_i$ at level $j$ learns against the opponent restricted policy set $\Pi^{r,<j}_{-i}\in\{\Pi^{r,k}_{-i}|k=1,\dots,j-1\}$ which is composed of active and fixed policies lower than $j$. Once a lowest active policy meets the stop condition (e.g. number of training episodes), it will be fixed to expand the corresponding restricted policy set, and a new active policy will be add into the queue with the highest level. We argue that parallel training will increase exploration efficiency since each best response plays against multiple policy sets.

\section{Experiments}\label{sec:exp}
We compare \framework{} with five algorithms, including Self-play~\cite{hernandez2019generalized}, PSRO~\cite{lanctot2017unified}, Rectified PSRO (PSRO-rN)~\cite{balduzzi2019open}, Mixed-Oracles~\cite{smith2020iterative}, and Pipeline PSRO (P-PSRO)~\cite{mcaleer2020pipeline}. The environments for the test are three classes of increasing difficulty games, i.e., matrix games, Poker games, and Multi-agent Gathering. We investigate the \emph{exploration efficiency} and \emph{computation efficiency} of algorithms in these experiments, also the performance on convergence. We use \textsc{NashConv} to evaluate the convergence quality in matrix games and Poker games and \emph{cardinality of payoff matrix}~\cite{regin2004cardinality} to evaluate the exploration efficiency. Since traversing the game tree of the Multi-agent Gathering game is too expensive, we set the policy sets $\Pi^{\text{PSRO}}$ produced by PSRO as the baseline to calculate the score of algorithms, which reflects the performance to some extent. The matrix games are designed for the comparison of \emph{exploration efficiency}. Especially the non-transitive mixture game (Section~\ref{sec:non_trans}), which vividly characterizes the exploration dynamics of algorithms. Furthermore, to demonstrate the importance of parallel training for exploration efficiency, we remove the pipeline URR solver of \framework{} as N\framework{} in the matrix games. As for the \emph{computation efficiency}, we investigate it from the number of samples and time consumption in Poker games and Multi-agent Gathering. Extra results and the pseudo-code of score calculation in Appendix~\ref{appendix:results}. All experiments were performed on a single machine with 64 CPUs, 256 G RAM, and 2 GeForce RTX 3090 GPUs.

\begin{figure*}[h]
	\centering
	\subfigure[NashConv]{\label{fig:alpha_star_nash}\includegraphics[width=.24\textwidth]{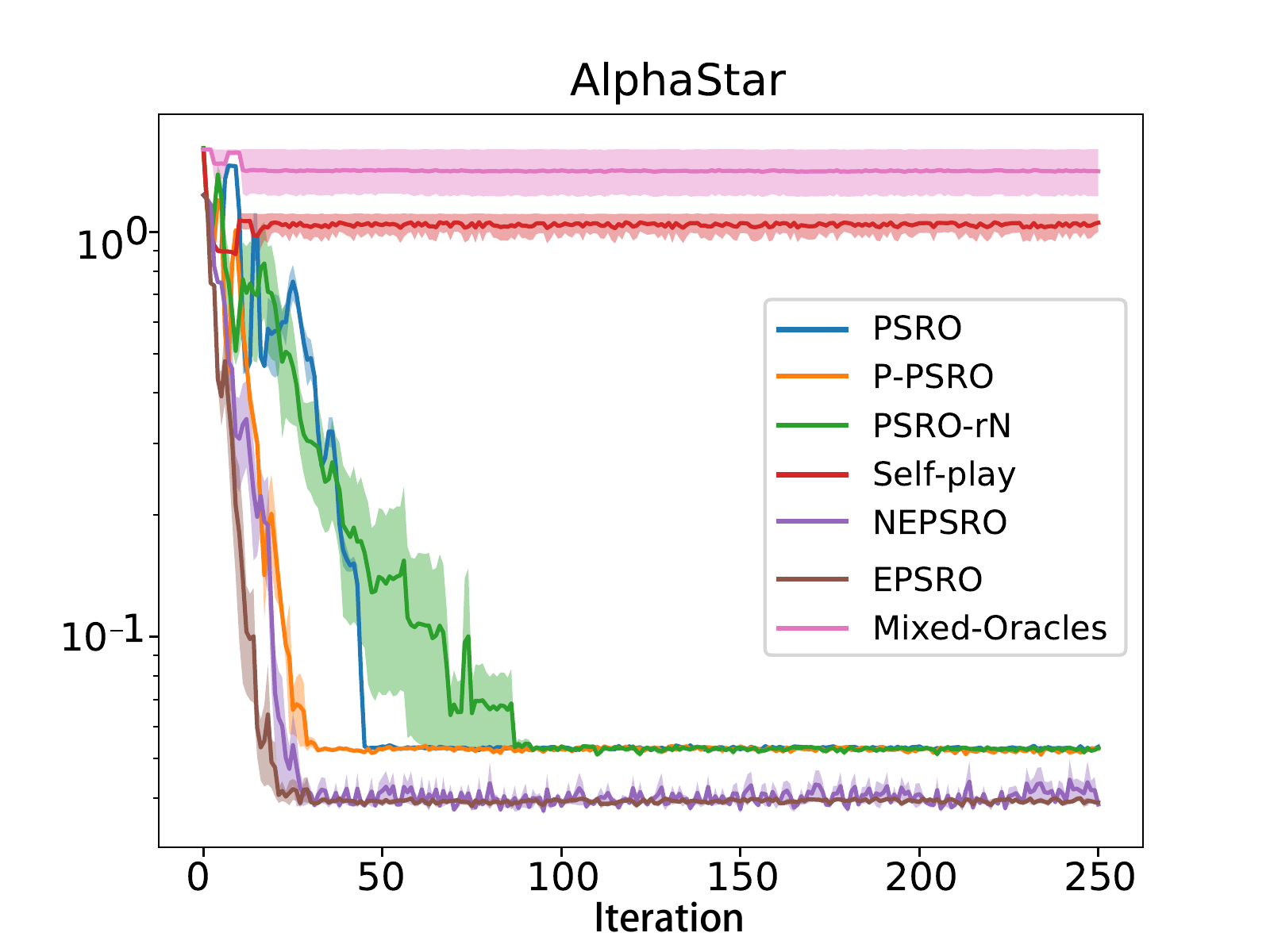}}
	\subfigure[Cardinality]{\label{fig:alpha_star_card}\includegraphics[width=.24\textwidth]{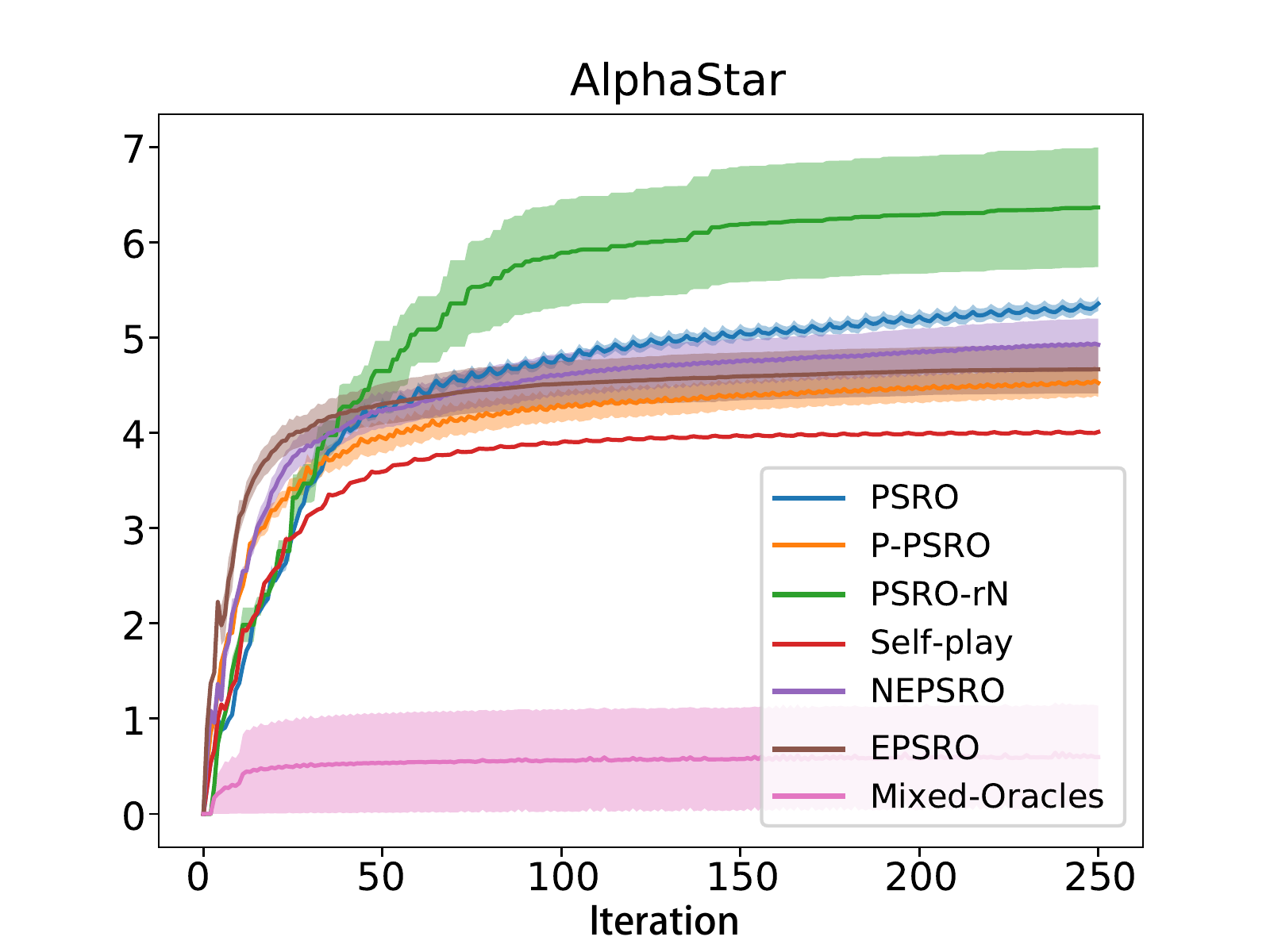}}
    \subfigure[NashConv ]{\label{fig:random_120_nash}\includegraphics[width=.24\textwidth]{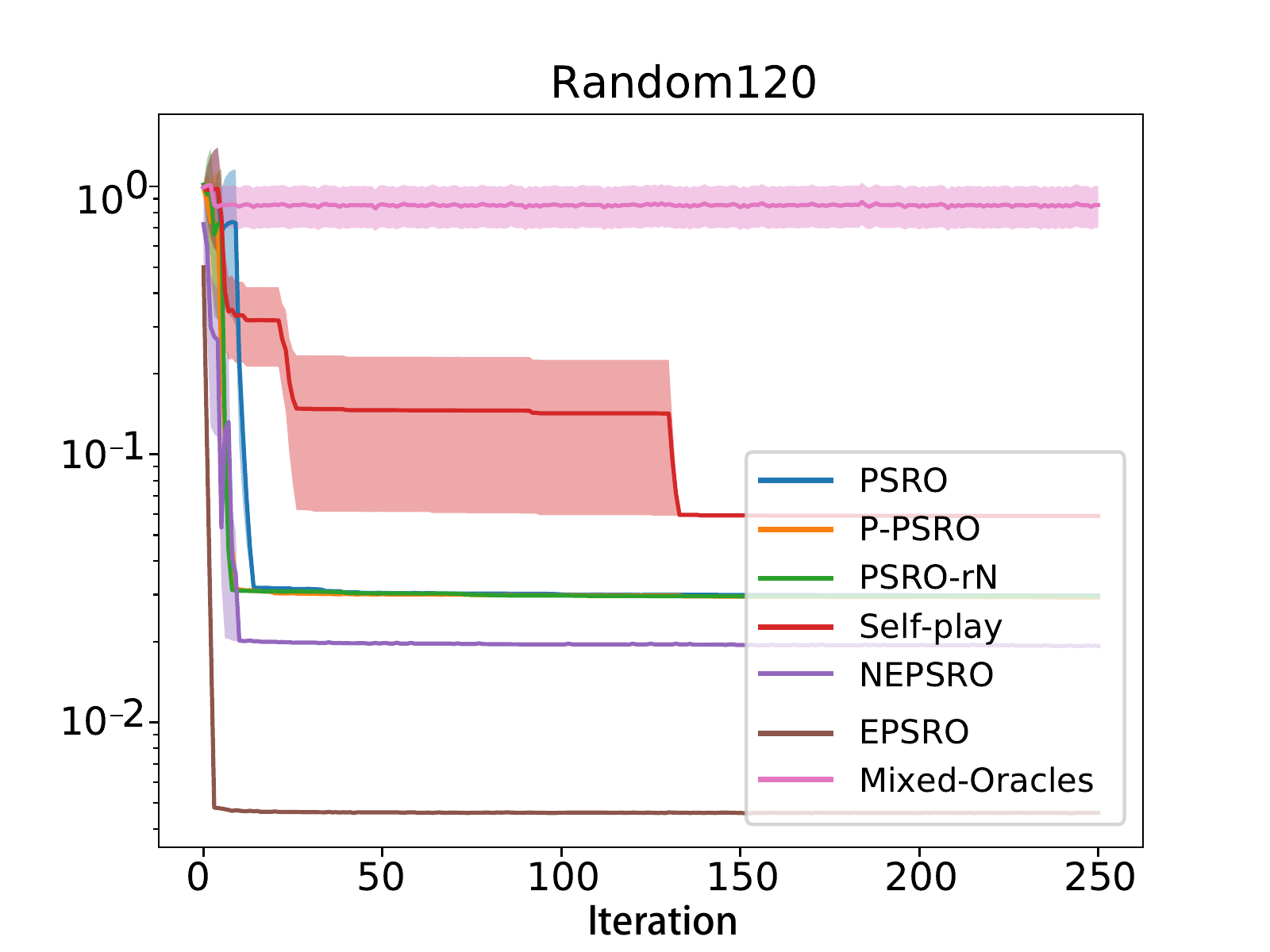}}
	\subfigure[Cardinality ]{\label{fig:random_120_card}\includegraphics[width=.24\textwidth]{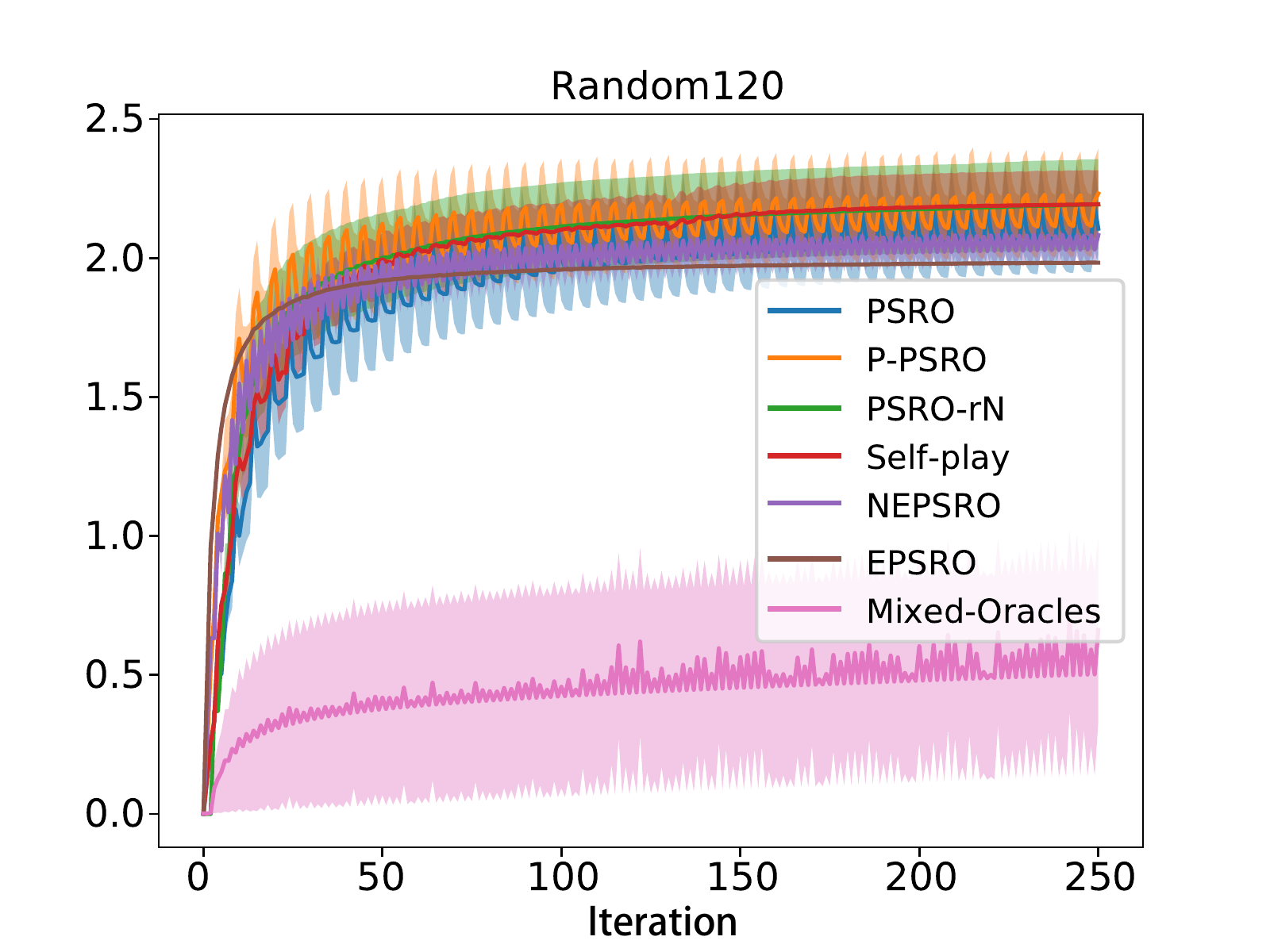}}
    \caption{Comparison of \textsc{NashConv} and cardinality. (a) and (b) show the \textsc{NashConv} and cardinality on AlphaStar matrix game, respectively; (c) and (d) show the \textsc{NashConv} and cardinality on a high-dimensional symmetric game ($n=120$), respectively. Though \framework{} has lower cardinality than some other algorithms, it outperforms all baselines on the \textsc{NashConv}. We argue that a lower cardinality but better exploitability indicates that the algorithm has higher exploration efficiency since it achieves better performance with fewer policies. More results in Appendix~\ref{appendix:results_random}.}
    \label{fig:high_dimensional}
\end{figure*}

\begin{figure*}[ht!]
	\centering
	\includegraphics[width=\linewidth]{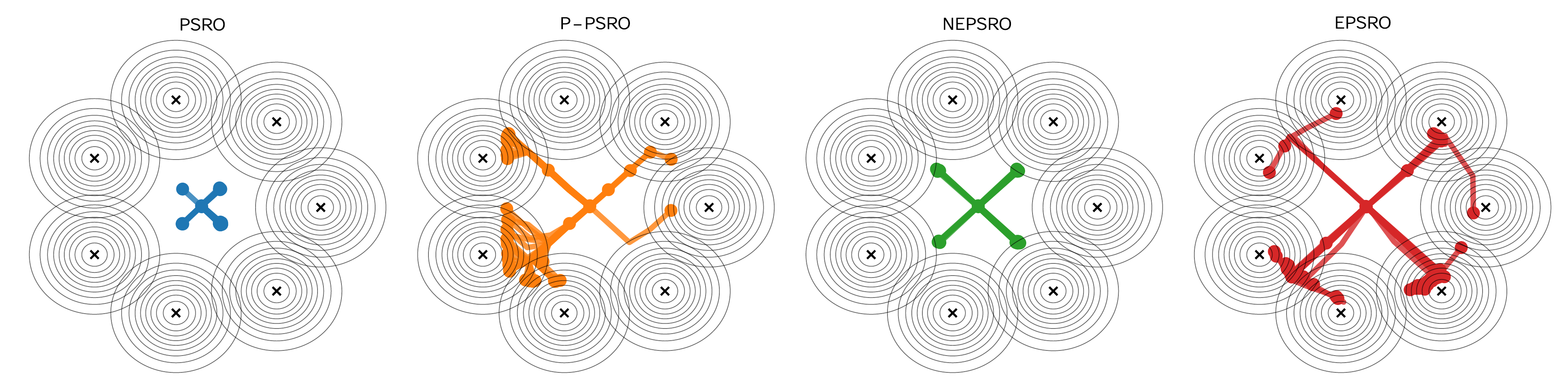}
    \caption{Exploration trajectories on \textit{Non-transitive Mixture Games}. The more trajectories close to the centers of Gaussian, the higher the exploration efficiency of the algorithm. Our algorithms (\framework{} and N\framework{}) outperform all selected baselines. Especially the \framework{}, it explored all centers. More results in Appendix~\ref{appendix:results_non_trans}.}
    \label{fig:non_mixture_exp}
\end{figure*}

\begin{figure*}[ht!]
	\centering
	\subfigure[vs. Wall-time]{\label{fig:kuhn_wall_time}\includegraphics[width=.24\textwidth]{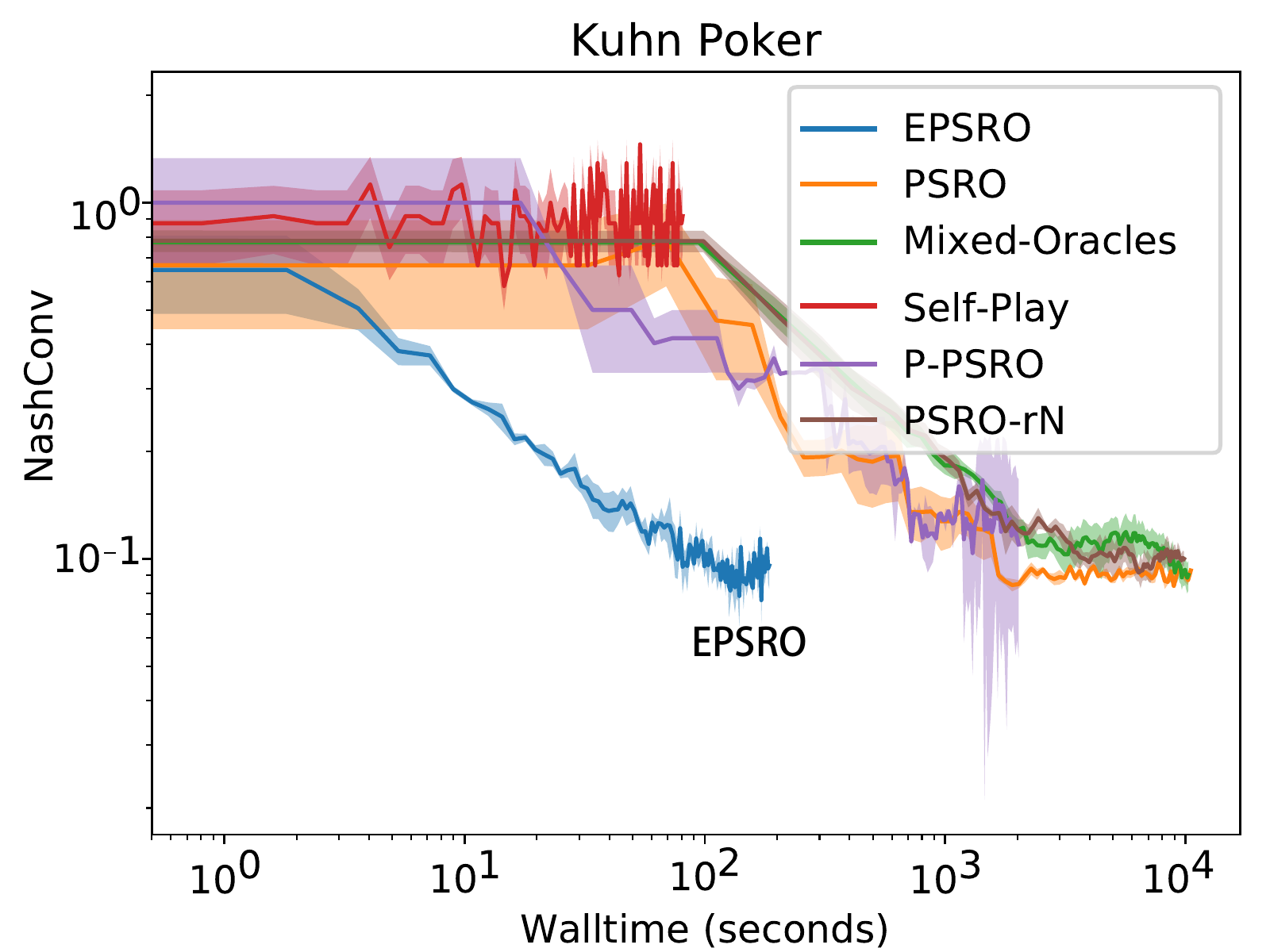}}
	\subfigure[vs. Samples Poker]{\label{fig:kuhn_sample}\includegraphics[width=.24\textwidth]{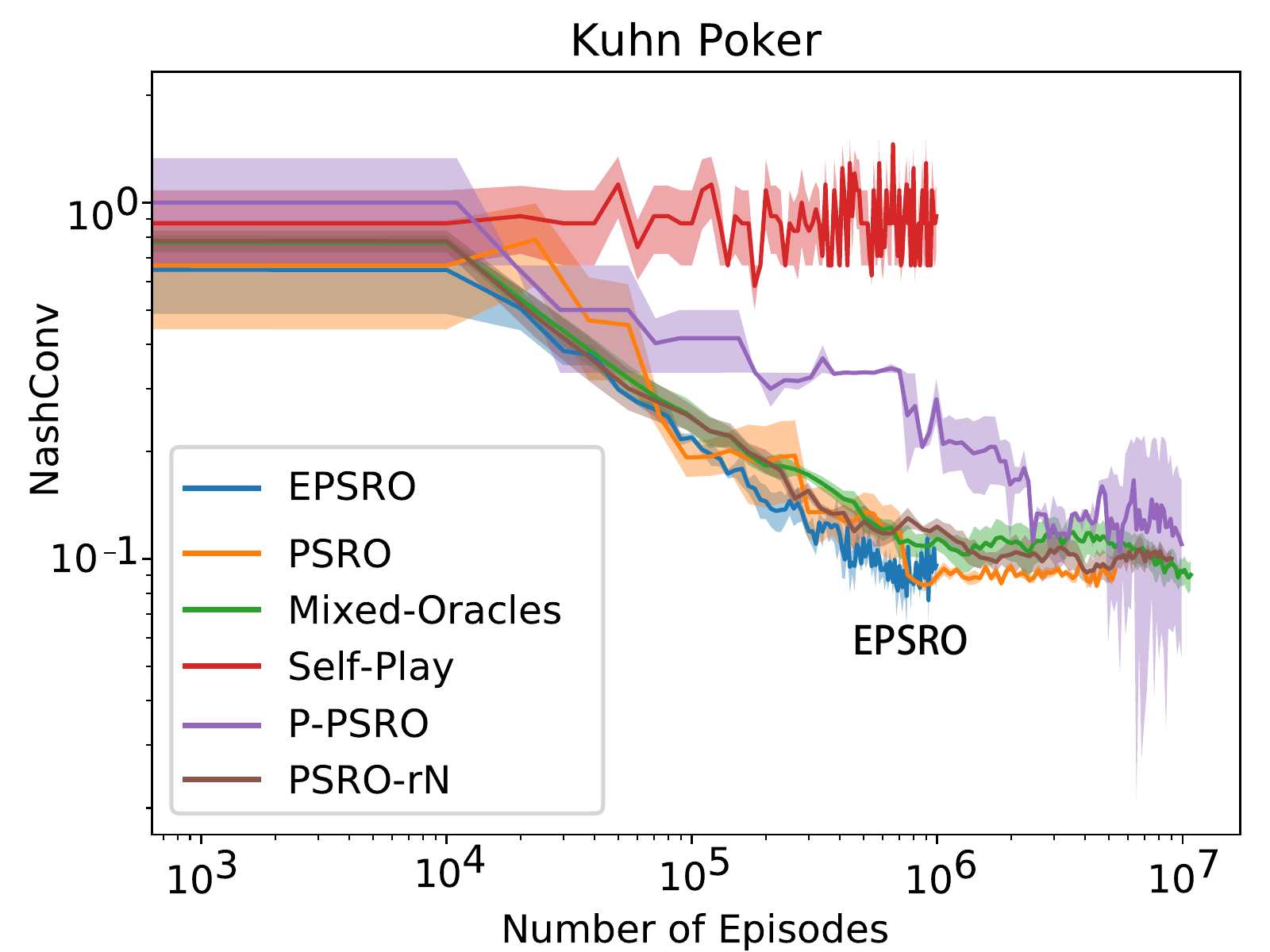}}
    \subfigure[vs. Wall-time Poker]{\label{fig:leduc_wall_time}\includegraphics[width=.24\textwidth]{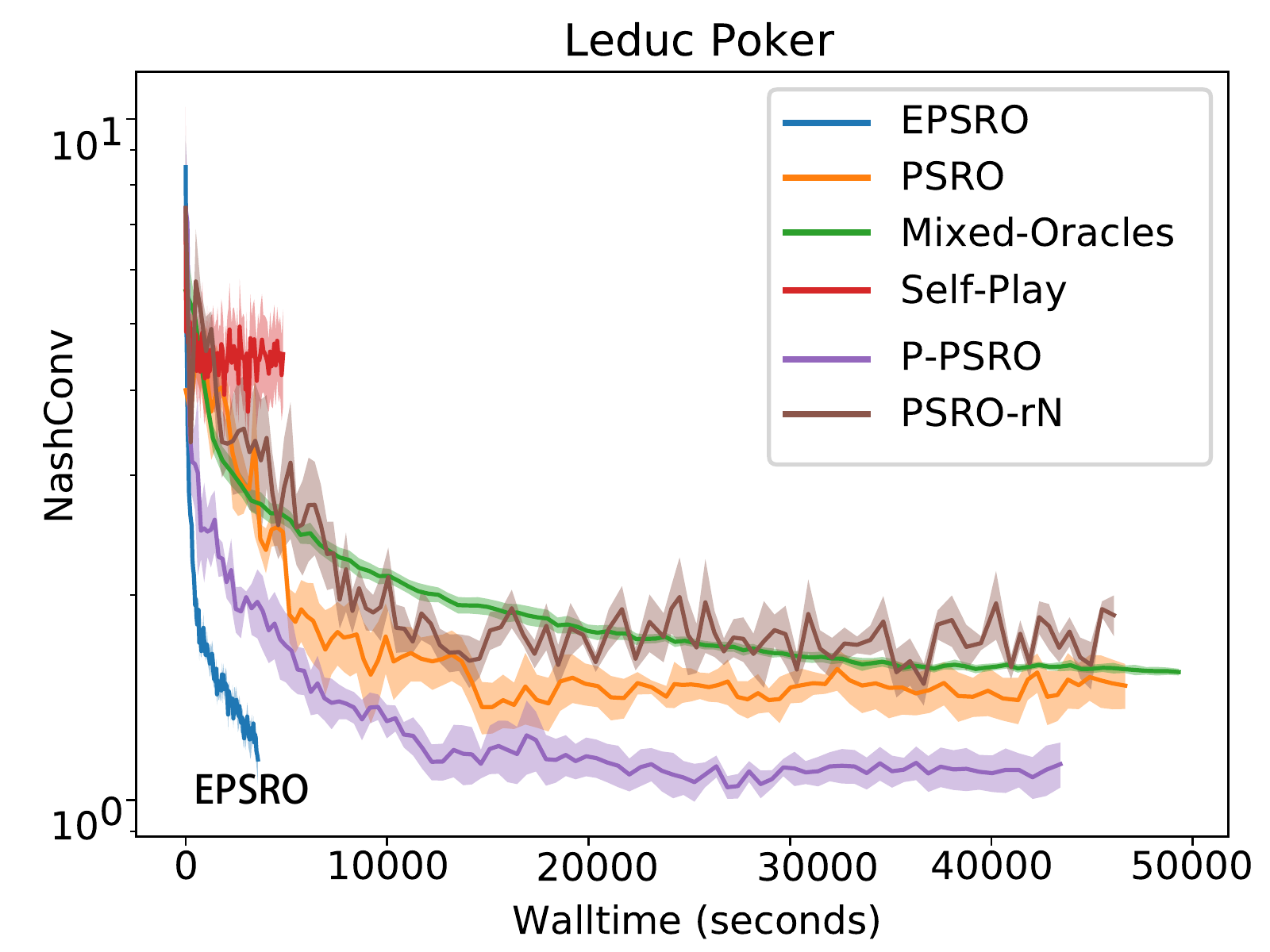}}
    \subfigure[vs. Samples Poker]{\label{fig:leduc_sample}\includegraphics[width=.24\textwidth]{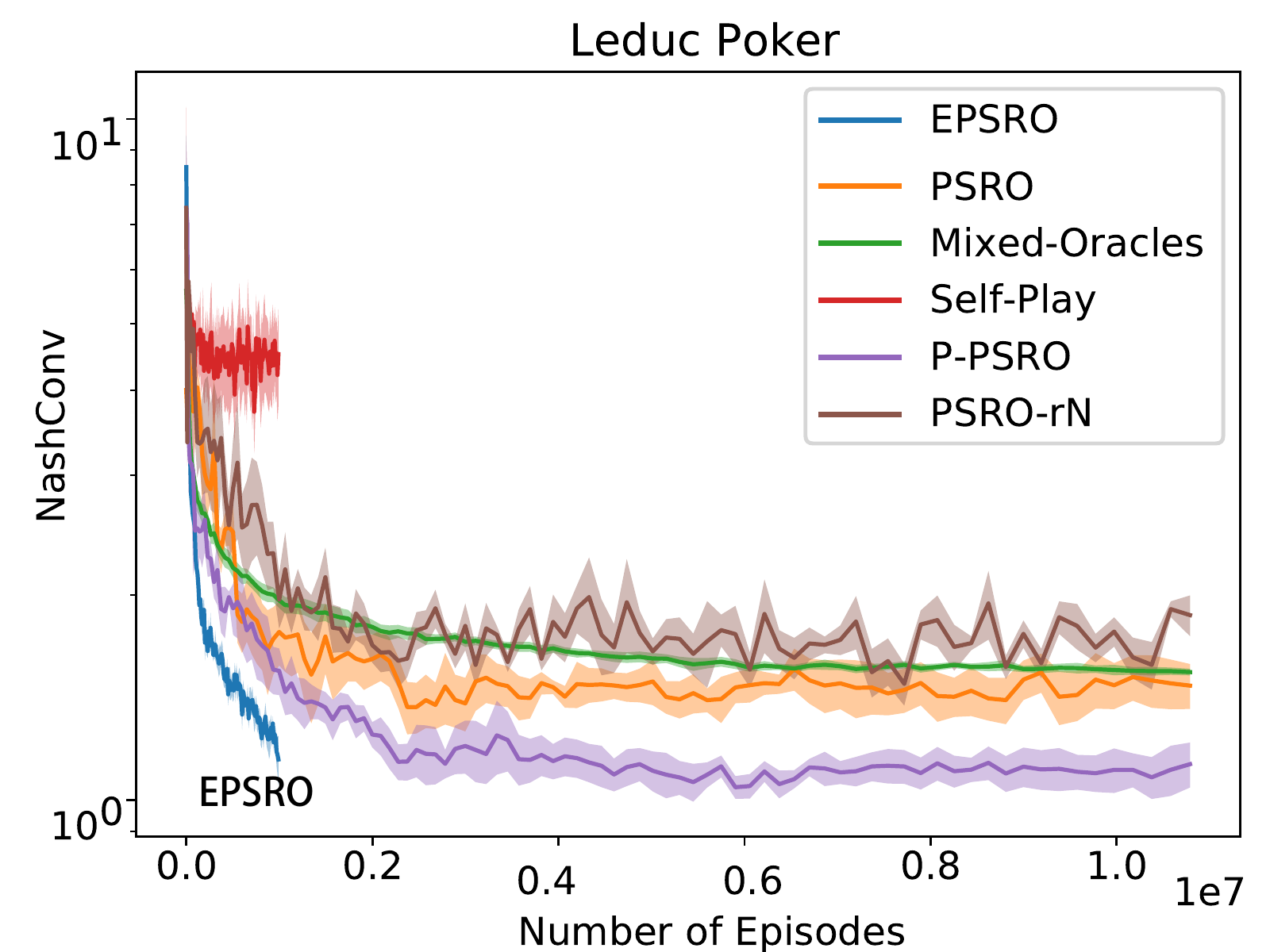}}
    \caption{NashConv on poker games. The number of samples for training each best response is set to $1e4$ episodes, and the number of simulations (if needed) for each joint policy is set to $1e3$ episodes. \framework{} performs a similar performance on the \textsc{NashConv}. For the computation efficiency, \framework{} achieves more than 50x seepup on wall-time; more than 10x sample efficiency than other algorithms.}
    \label{fig:poker_games}
\end{figure*}

\subsection{Comparison of Exploration Efficiency}\label{sec:non_trans}
The non-transitive game for the comparison of \emph{exploration efficiency} is a zero-sum two-player game consisting of 7 equally-distanced Gaussian humps on the 2D plane. In this game, each player chooses a point in the 2D plane as its decision, which is transformed into a 7-dimensional vector $\pi_i$ with each coordinate being the density in the corresponding Gaussian distribution. The payoff of the game is given by $\phi_i(\pi_i, \pi_{-i}) = \pi^T_i\textbf{S}\pi_{-i} + \textbf{1}^T(\pi_i - \pi_{-i})$. In this game, the optimal strategy should stay close to the center of the Gaussian and explore all the Gaussian distributions equally. We train best response policies in 50 epochs for each algorithm.
As presented in Figure~\ref{fig:non_mixture_exp}, \framework{} successfully explores all the centers and shows higher \emph{exploration efficiency} than other baselines. Though the NEPSRO fails to achieve to any centers, its explored policy space is larger than most algorithms.

\subsection{High-dimensional Matrix Games}
The high-dimensional matrix games introduced here include two classes. One is a symmetric matrix game generated in a high-dimensional uniform distribution, and another is an empirical payoff matrix corresponding to 888 reinforcement learning policies in AlphaStar~\cite{vinyals2019grandmaster}. We demonstrate the comparison of these games on the performance when dealing with games that have high-dimensional policy spaces.
\paragraph{Random Symmetric Matrix Game.}
\cite{mcaleer2020pipeline} introduce the games to investigate the performance of PSRO-based methods in high-dimensional symmetric games (SymGame). In this experiment, we generated random symmetric zero-sum matrices with different dimension $n$. For a given matrix, elements in the upper triangle are distributed uniformly: $\forall i < j \le n, a_{i,j} \sim \textsc{Uniform}(-1,1)$ and for the lower triangle, the elements are set to be the negative of its diagonal counterpart: $\forall j <i \le n, a_{i,j} = -a_{j,i}$. The diagonal elements are equal to zero: $a_i,i = 0$. The matrix defines the utility of two pure strategies to the row player. In these experiments, we train a strategy $\pi$ as a best response that plays against another strategy $\hat{\pi}$, it is updated by a learning rate $r$ multiplied by the best response to that strategy: $\pi' = r\textbf{BR}(\hat{\pi}) + (1-r)\pi$. Figure~\ref{fig:high_dimensional} shows the results for $n=120$. We report both \textsc{NashConv} and cardinality. The results show that \framework{} and NEPSRO achieve a faster convergence rate and the lowest \textsc{NashConv} than all of the other algorithms. Though they do not achieve the highest cardinality, we argue that there is a tradeoff between convergence and exploration, and \framework{} performs a better balance between them. It is worth mentioning that the Mixed-Oracle fails to seek a meta-strategy that has a smaller distance to the Nash equilibrium, even worse than Self-Play. We argue that the policy distills of Mixed-Oracles may decrease the exploration efficiency, especially in such a high-dimensional policy space.
\paragraph{AlphaStar Empirical Game.}
The AlphaStar Matrix Game is derived from solving a complex real-world game StraCraftII~\cite{czarnecki2020real}, which involves 888 reinforcement learning policies. We test the \emph{exploration efficiency} and the convergence quality of our method for solving such empirical games. Similar to the results in the random symmetric matrix game, our algorithm performs a faster convergence rate and lower \textsc{NashConv} than other algorithms, while the Mixed-Oracle still fails to explore new policies to expand its policy sets (Figure~\ref{fig:high_dimensional}).

\subsection{Poker Games}
\begin{wrapfigure}[11]{r}{0.45\linewidth}
	\vspace{-4em}
	\begin{center}
		\includegraphics[width=\linewidth]{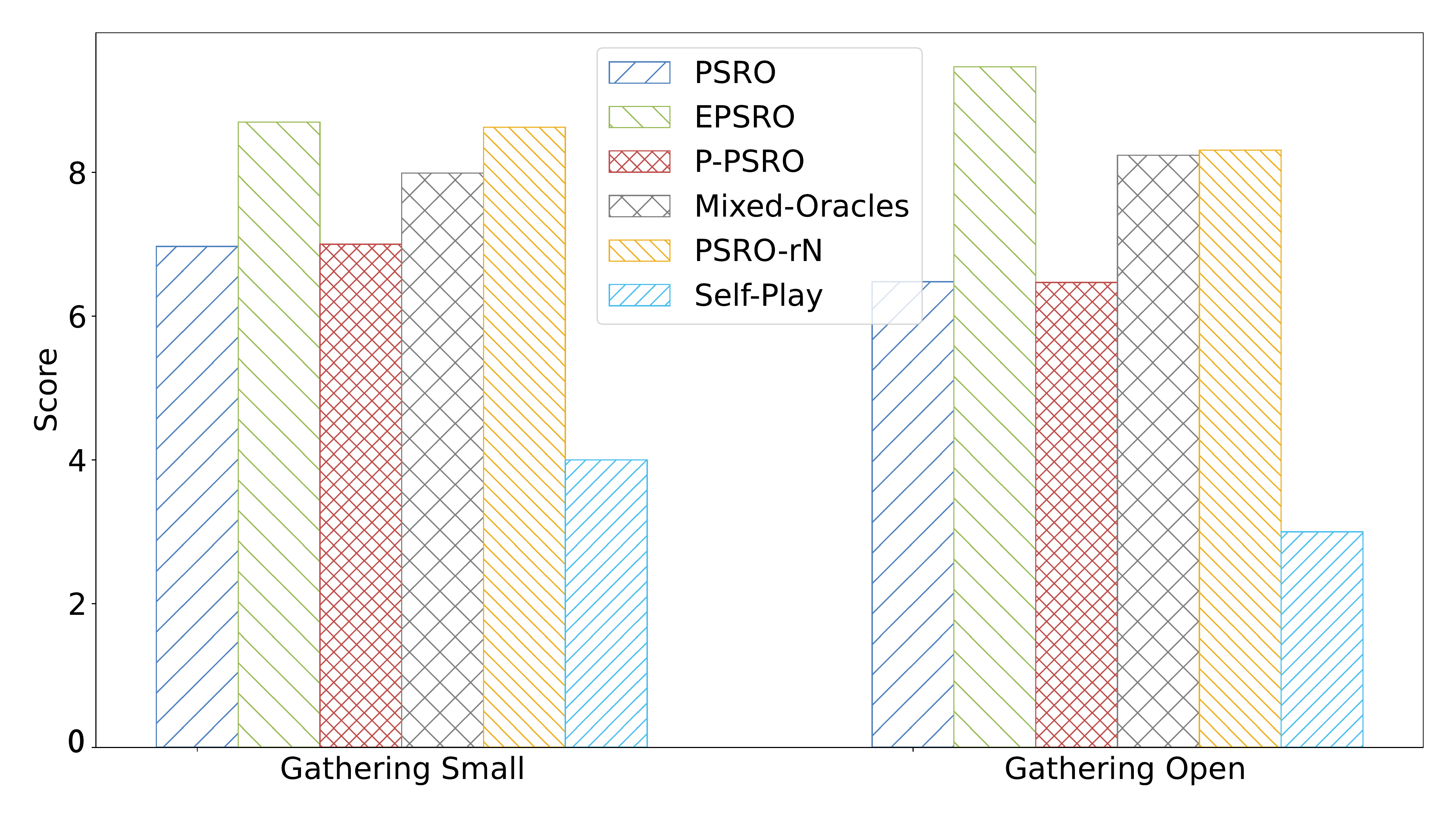}
	\end{center}
	\caption{The average score of final policy set that plays against $\Pi^{\text{PSRO}}$.}
	\label{fig:gathering_score}
\end{wrapfigure}
Poker is a common benchmark in multi-agent decision tasks. In this paper, we introduce two simplified forms of poker games for experiments, i.e., Kuhn Poker and Leduc Poker~\cite{zinkevich2007regret,bowling2015heads}.
These poker tasks model zero-sum two-player imperfect-information games, in which each player shows uncertainty about the game rules and the state of other players.
Similar to Poker, where each player in these games chooses to raise/call/fold through rounds of betting. We investigate the sample efficiency and performance of \framework{} in this experiment. The training for each algorithm is set to learn 100 policies.
The number of samples for training each policy is set to 10000 episodes, and the number of simulations for each joint policy is set to 1000 episodes. So the total number of samples for training a PSRO algorithm with simulations achieves $1.1 \times 1e7$. Figure~\ref{fig:poker_games} presents the results of \textsc{NachConv} w.r.t to wall-time and the number of samples.
Since no simulations and parallel best response learning, \framework{} substantially achieves high training efficiency.
Specifically, for the wall-time, \framework{} has more than 50x speedup than other PSRO methods; for the sample efficiency, \framework{} has more than 10x than other non-parallelized PSRO methods. 


\subsection{Multi-agent Gathering}
We further investigate the capability of \framework{} to handle more complex multi-agent tasks in multi-agent gathering environments (MAG)\footnote{\url{https://github.com/HumanCompatibleAI/multi-agent}}.
MAG is an endless environment whose horizon length could be infinite.
In our experiments, we limit the horizon of each episode to 100. In MAG, the goal of each agent is to collect as many as ``apples''.
The apples regrow at a rate dependent upon the configuration of the uncollected nearby apples. In this case, the more nearby apples, the higher the regrowth rate of the apples.
Naturally, this presents a dilemma for the players: each wants to pick as many apples as possible.
However, if they over-harvest the throughput of apples diminishes, potentially falling to zero. Figure~\ref{fig:gathering_score} shows the confrontation results of each algorithm's final policy set $\Pi^{\text{TEST}}$ to an evaluation policy set $\Pi^{\text{PSRO}}$. We calculate the score as $\sigma^{\text{TEST}}M^{\text{TEST}}\left[\sigma^{\text{PSRO}}\right]^T$, where $M^{\text{TEST}}$ is the empirical payoff matrix and $\sigma$ is a learned meta-strategy. The higher the score, the better the performance of the corresponding algorithm. Additionally, we report the curve of the score of intermediate policy sets during the learning process in Appendix~\ref{appendix:results_gathering}.

\section{Related Work}
In large normal-form games, it is difficult to directly compute the approximation of Nash equilibrium.
Policy Space Response Oracles (PSRO)~\cite{lanctot2017unified} provides an iterative solution to solve this problem. There are many variants of PSRO focus to improve the convergence rate or training efficiency.
A straightforward idea is to utilize parallelism to improve the training efficiency. For instance, DCH~\cite{lanctot2017unified} parallelizes PSRO by training multiple RL policies, each against the meta Nash equilibrium below it in the hierarchy. A problem with this method is that the number of policies should be set beforehand. However, it is difficult to figure out how many policies does it require to solve a game in practice. \cite{mcaleer2020pipeline} proposed a similar solution, i.e. P-PSRO, to solve this problem, which inherits the hierarchical parallelism training but has no need to preset the number of training policies. Another parallelism variant is Rectified PSRO~\cite{balduzzi2019open}, but it has been proved not converge in all symmetric zero-sum games~\cite{mcaleer2020pipeline}.

Another factor to the efficiency is the exploration efficiency. Specifically, more diverse the learned best responses, closer the restricted sets to original policy sets~\cite{perez2021modelling,yang2021diverse}. Thus, PSRO-based methods converge in smaller size of restricted policy sets and less time consumption. However, it is difficult to discover an exact best response, because the underlying policy training has no guarantee to find an ideal policy as the best response without extra conditions. As a solution, improving the diversity of restricted policy set has been regarded as a reasonable way to solve this problem. Among the existing work, DPP~\cite{perez2021modelling} utilizes the expected cardinality to measure the diversity of policy set. \cite{liu2021towards} proposed a method that unify both behavior and reward distance to measure the diversity. Since the learning of best response need to play against a mixture of opponent policies. Many of existing work demonstrate by sampling opponent policies from this mixture for each episode. Compared to playing against a single policy, \cite{smith2020iterative} claimed that such a mechanism brings stochasticity on the opponent and forgets previous experiences, making algorithms slow to converge. Thus, the authors proposed a method that distills~\cite{czarnecki2019distilling} the opponent mixture as a single policy via Q-mixing~\cite{davydov2021q}.

Concurrently to our work, \cite{mcaleer2022anytime} proposed a similar work named AODO. Their work differs from ours in the following ways: (1) We focus on improving the training efficiency of PSRO-based methods, providing monotone improvement and convergence analysis for \framework{}, while AODO focuses on making a guarantee to monotonic decrease the \textsc{NashConv}~\cite{johanson2011accelerating}; (2) \framework{} introduces parallelized best response training while AODO executes in singleton; (3) We proposed a warm-start technique to tackle the re-training problem while AODO starts from scratch at each epoch; (4) We demonstrate the experiments with 5 baselines on both high-dimensional matrix games, poker games, and non-trivial multi-agent gathering tasks while AODO considers only one baseline and fewer environments.

\section{Conclusions}\label{sec:conclusion}
In this work, we propose a parallel algorithm \framework{} to improve the training efficiency of PSRO.
The demonstration results show that \framework{} achieves higher computation efficiency and exploration efficiency than existing works. The improvements of \framework{} benefit from learning best responses against the whole opponent restricted policy set and cooperating with parallelized training. In addition, a warm-start technique to reduce the re-training cost also makes our algorithm perform a higher convergence rate.
However, \framework{} is limited to handling the two-player cases because there will be a divergence in selecting meta-strategies for more players involved. In future work, we would like to seek a method to solve this problem and generalize EPSRO to multi-player cases.

\bibliographystyle{plain}
\bibliography{refs}

%
%


\newpage
\appendices


\section{Algorithms}\label{appendix:algorithm}

Algorithm~\ref{alg:deterministic} shows the pseudo-code of meta-strategy optimization.

\begin{algorithm}[h]
	\caption{\textsc{Deterministic Meta Strategy Optimization}}\label{alg:deterministic}
	\KwInput{initialize the $\sigma_{-i,\beta}$ with Algorithm~\ref{alg:warm_start}}
	\KwInput{initialize an episode reword buffer $\mathcal{B}$;window size $L$; $K$ the length of meta-strategy $\sigma_{-i,\beta}$}
	\KwInput{loss vector buffer $\mathcal{L}_i, \mathcal{L}_{-i} \in \mathbb{R}^{K}$; counters: $N_1=0,\dots,N_K=0$}
	\While{not terminated}{
		\For{many episodes}{
			Sample $\pi^k_{-i} \sim \sigma_{-i,\beta}$ to play against $\pi_{i,\theta}$
			
			Collect episode return $r^k$ into $\mathcal{B}$ and $N_k := N_k + 1$
			
			\If{the size of $\mathcal{B}$ equals to $L$}{
				Calculate average observed returns: $\forall k=1,\dots,K, \bar{r}^k=\frac{\sum^L_{l=1}\textbf{1}_{j=k}r^j_l}{N_k}$
				
				Compute gradients for $\beta_k=\arg\max_{\beta_k}\bar{r}^k$ and update $\sigma_{-i,\beta}$ with Lemma~\ref{lemma:beta_opt_b}

				Collect loss vectors $l^k_{-i}=-\bar{r}^k$, $l^k_i=\bar{r}^k$ into $\mathcal{L}_i$, $\mathcal{L}_{-i}$, respectively
				
				Reset buffer $\mathcal{B}$ and counters
			}
		}
	}
	\KwOutput{meta-strategy $\sigma_{-i}$ for player $-i$, loss vector buffers $\mathcal{L}_i$, $\mathcal{L}_{-i}$}
\end{algorithm}

We introduce the warm-start for opponent meta-strategy as follows. For each player $i$, the algorithm randomly initializes a vector $\beta$ with the length equals to the size of $\Pi^{r}_{-i}$. At each epoch $e+1$, we slightly run a simulation procedure of $(\pi_i, \pi_{-i})$ to estimate the utility $u_{-i}(\pi_i,\pi_{-i})$ before applying the Lemma~\ref{lemma:beta_opt} to optimize $\beta$. Then, by cooperating with the average loss vector from last epoch $e$, we could start the optimization to compute a feasible $\beta^*$ that makes $\sigma_{-i}$ satisfy Theorem~\ref{theorem:averaging} and~\ref{theorem:warm_bound}.

\begin{algorithm}[h]
	\caption{\textsc{Meta-Strategy Warm-Start}}\label{alg:warm_start}
	\KwInput{newly learned policy support $\pi_{-i}$; best response $\pi_i$ from last epoch; error threshold $\tau \ge 0$}
	\KwInput{initialize $\sigma_{-i}$ with $\beta\in \mathbb{R}^{|\Pi^r_{-i}|}$; loss vectors from last epoch $[l^1_{-i},\dots,l^T_{-i}]$}
	
	Compute average loss vector $\bar{l}_{-i}=\left(l^1_{-i}+\dots+l^T_{-i}\right) / T$

	Estimate $u_{-i}(\pi_i, \pi_{-i})$ by running simulation for $(\pi_i, \pi_{-i})$
	
	Compute the square error of utility as $\xi_2 = \parallel (x-\sigma_{-i})^{\top} \bar{l}_{-i} - \sigma_{-i}(k) u_{-i}(\pi_i,\pi_{-i})\parallel_2 - \lambda H(\sigma_{-i})$
	
	Update $\beta$ as $\beta:=\beta - \nabla_{\beta}\xi_2$
	
	\If{$\xi_2 \le \tau$}{Stop optimization and continue}
	\Else{Go back to step 3}
	
	\KwOutput{meta-strategy $\sigma_{-i}$}
\end{algorithm}

Algorithm~\ref{alg:warm_start} relies on the computation of $\xi_2$, we introduce its definition in Lemma~\ref{lemma:beta_opt_b}.

\section{Proofs}\label{appendix:proofs}

\addcontentsline{toc}{subsection}{A Proof of Theorem~\ref{theorem:weak_strength}}
\subsection*{A Proof of Theorem~\ref{theorem:weak_strength}}

\begin{theorem}[\emph{Monotonic Policy Space Expanding}]\label{theorem:weak_strength_b}
	For any given epoch $e$ and $e+1$, let $(\pi^e_i, \sigma^e_{-i})$ and $(\pi^{e+1}_i, \sigma^{e+1}_{-i})$ be Nash equilibrium of $\textbf{URR}^e_i$ and $\textbf{URR}^{e+1}_i$, respectively, where $\pi^e_i,\pi^{e+1}_i \in \Pi_i$, $\sigma^e_{-i} \in \Delta^e_{\Pi^{r}_{-i}}$ and $\sigma^{e+1}_{-i} \in \Delta^{e+1}_{\Pi^{r}_{-i}}$.
	The utilities of $\pi^e_i$ against opponent strategies $\sigma^e_{-i}$ and $\sigma^{e+1}_{-i}$ satisfies
	\begin{equation}
		u_i(\pi^e_i, \mathbf{\sigma^e_{-i}}) - u_i(\pi^e_i, \mathbf{\sigma^{e+1}_{-i}}) \ge 0,
	\end{equation}
	where $\Delta^e_{\Pi^r_{-i}}$ indicates $\Delta(\Pi^{r,e}_{-i})$. Especially, $u_i(\pi^e_i,\sigma^e_{-i}) - u_i(\pi^e_i,\sigma^{e+1}_{-i}) > 0$ indicates there is a strictly policy space expanding at $e+1$, i.e., $\pi^{e+1}_{-i} \in \Pi^{r,e+1}_{-i}\setminus\Pi^{r,e}_{-i}$.
\end{theorem}
\begin{proof}
	Note that the proof is non-trivial since $\Pi^{r,e}_{-i} \subset \Pi^{r,e+1}_{-i}$, so we cannot directly derive Theorem~\ref{theorem:weak_strength_b} with the Nash equilibrium $(\pi^e_i,\sigma^e_{-i})$. Instead, we should combine equilibriums of epoch $e$ and $e+1$. Additionally, the equilibrium considered is mixed-strategy Nash equilibrium. Considering the property of Nash equilibrium, $\forall \pi_i \in \Delta(\Pi_i), u_i(\pi_i, \sigma^{e+1}_{-i}) \le u_i(\pi^{e+1}_i, \sigma^{e+1}_{-i})$, then we have
	\begin{equation}\label{eq:tt_puls_1}
		u_i(\pi^e_i, \sigma^{e+1}_{-i}) \le u_i(\pi^{e+1}_i, \sigma^{e+1}_{-i}).
	\end{equation}
	Analogously, $\forall \sigma_{-i} \in \Delta^{e+1}_{\Pi^r_{-i}}, u_i(\pi^{e+1}_i, \sigma^{e+1}_{-i}) \le u_i(\pi^{e+1}_i, \sigma_{-i})$, then we have
	\begin{equation}\label{eq:t_puls_1_t}
		u_i(\pi^{e+1}_i, \sigma^{e+1}_{-i}) \le u_i(\pi^{e+1}_i, \sigma^e_{-i}).
	\end{equation}
	Combining Eq.~(\ref{eq:tt_puls_1}) and~(\ref{eq:t_puls_1_t}), we can derive
	\begin{align}
		u_i(\pi^e_i, \sigma^e_{-i}) &- u_i(\pi^e_i, \sigma^{e+1}_{-i}) \\\nonumber
		&\ge u_i(\pi^e_i, \sigma^e_{-i}) - u_i(\pi^{e+1}_i, \sigma^{e+1}_{-i})\\\nonumber
		&\ge u_i(\pi^e_i, \sigma^e_{-i}) - u_i(\pi^{e+1}_i, \sigma^e_{-i}).
	\end{align}
	Since $(\pi^e_i, \sigma^e_{-i})$ is a Nash equilibrium, then there is $u_i(\pi^e_i, \sigma^e_{-i}) - u_i(\pi^{e+1}_i, \sigma^e_{-i}) \ge 0$.
	Thus $u_i(\pi^e_i, \sigma^e_{-i}) - u_i(\pi^e_i, \sigma^{e+1}_{-i}) \ge 0$. Then we give a proof of strict expanding as Corollary~\ref{corollary:strict_expanding}. 
\end{proof}


\begin{corollary}[\emph{Strict Expanding}]\label{corollary:strict_expanding}
	For equilibrium $(\pi^e_i,\sigma^e_{-i})$ and $(\pi^{e+1}_i,\sigma^{e+1}_{-i})$, there is a strict expanding when $u_i(\pi^e_i,\sigma^e_{-i}) > u_i(\pi^e_i,\sigma^{e+1}_{-i})$ and $\Pi^{e,r}_{-i} \subset \Pi^{e+1,r}_{-i}$. 
\end{corollary}
\begin{proof}
	Suppose $\sigma^{e+1}_{-i}\in\Delta^e_{\Pi^r_{-i}}$ when $u_i(\pi^e_i,\sigma^e_{-i}) - u_i(\pi^e_i,\sigma^{e+1}_{-i}) > 0$ holds. As $(\pi^{e+1}_{i},\sigma^{e+1}_{-i})$ is also an equilibrium of $(\Pi_i, \Pi^{e,r}_{-i})$, so there is $u_i(\pi^e_i,\sigma^e_{-i})=u_i(\pi^{e+1}_i,\sigma^{e+1}_{-i})$. Thus we have
	\begin{equation}
		u_i(\pi^e_i,\sigma^e_{-i}) - u_i(\pi^e_i,\sigma^{e+1}_{-i}) \ge u_i(\pi^e,\sigma^e_{-i}) - u_i(\pi^{e+1}, \sigma^{e+1}_{-i}) = 0,
	\end{equation}
	which violates the inequality. So there must be $\sigma^{e+1}_{-i}\in\Delta^{e+1}_{-i}\setminus\Delta^e_{-i}$.
\end{proof}

%
\addcontentsline{toc}{subsection}{Discussion of Exploration Efficiency}
\subsection*{Discussion of Exploration Efficiency}\label{sec:higher_exploration}
We further discuss the exploration efficiency from the theoretical perspective as follows.

\begin{proposition}\label{prop:higher_exploration_b}
EPSRO has higher exploration efficiency than PSRO.
\end{proposition}
\begin{proof}
	Before the proof, we note that $\sigma_i \in \Delta(\Pi^r_i)$ and $\pi_i \in \Delta(\Pi_i)$ for each player $i \in \{1,2\}$. The \textsc{NashConv} (NC) of PSRO is
	\begin{align*}
		\textsc{NC}^{psro} &= \max_{\pi_i \in \Delta(\Pi_i)} u_i(\pi_i,\sigma_{-i}) + \max_{\pi_{-i} \in \Delta(\Pi_{-i})}u_{-i}(\sigma_i,\pi_{-i})\\
		&=\max_{\pi_i \in \Delta(\Pi_i)} u_i(\pi_i,\sigma_{-i}) - \max_{\pi_{-i} \in \Delta(\Pi_{-i})}u_{i}(\sigma_i,\pi_{-i})\\
		&=u_i(\pi^*_i,\sigma_{-i}) - u_i(\sigma_i,\pi^*_{-i}),
	\end{align*}
    and $u_i(\pi^*_i,\sigma_{-i}) \ge u_i(\sigma_i,\sigma_{-i}) \ge u_i(\sigma_i,\pi^*_{-i})$.
    
    At each epoch, \framework{} directly learns a tuple of strategies as an equilibrium for each player, i.e., $(\pi^*_i,\sigma^*_{-i})$ and $(\sigma^*_i,\pi^*_{-i})$ are two NEs of $(\Pi_i,\Pi^r_{-i})$, $(\Pi^r_i,\Pi_{-i})$, respectively. Thus, we can apply the same rule as PSRO to derive \framework{}'s \textsc{NashConv} which satisfies the following inequality:
    \begin{align*}
    	\textsc{NC}^{epsro} &= u_i(\pi^*_i,\sigma^*_{-i}) - u_i(\sigma^*_i,\pi^*_{-i})\\
    	&\le u_i(\pi^*_i, \sigma_{-i}) - u_i(\sigma^*_i,\pi^*_{-i}) \\
    	&\le u_i(\pi^*_i, \sigma_{-i}) - u_i(\sigma_i, \pi^*_{-i}) = \textsc{NC}^{psro}.
    \end{align*}
    Though $\textsc{NC}^{epsro} \le \textsc{NC}^{psro}$, the equation only holds when $(\pi^*_i, \sigma_{-i})$ and $(\sigma_i,\pi^*_{-i})$ are equilibrium because it requires $\Pi^r_i = \Pi_i$ for each player. Thus, $\textsc{NC}^{epsro} < \textsc{NC}^{psro}$ at each epoch before convergence, i.e. \framework{} has higher exploration than PSRO.
\end{proof}
\addcontentsline{toc}{subsection}{Modeling The Learning as MWU}
\subsection*{Modeling The Learning as MWU}
In this paper, we argue that the updates of meta-strategies (Algorithm~\ref{alg:deterministic}) follow the rule of Multiplicative Weights Update (MWU) algorithms. Before the explanation of equivalence, we introduce the definition of MWU as follows.

\begin{definition}[\emph{Multiplicative Weights Update~\cite{freund1999adaptive}}]
Given a game with utility matrix $U$ for the row player, let $c_1,c_2,\dots$ be a sequence of mixed strategies played by the column player. The row player is said to follow MWU if $\pi_{t+1}$ is updated as follows:
\begin{equation}
    \pi_{t+1}(i)=\pi_t(i)\frac{\exp{(\mu_t {a^i}^{\top}Uc_t)}}{\sum^n_{i=1}\pi_t(i)\exp{(\mu_t {a^i}^{\top}Uc_t)}}, \forall i \in [n]
\end{equation}
where $\mu_t > 0$ is a parameter, $\pi_0=[1/n,\dots,1/n]$ and $n$ is the number of pure strategies (a.k.a. experts).
\end{definition}

Then we prove the equivalence in the following Lemma.

\begin{lemma}\label{lemma:meta_update_b}
	Let $\beta=\left[\beta_1,\dots,\beta_n\right]$ be the weights vector of the meta strategy played by player $-i$. The player is said to follow MWU if $\sigma_{-i,\beta'}$ is updated as follows
	\begin{equation}
	    \sigma_{-i,\beta'}(j) = \frac{\exp{(\beta_j - g_j)}}{\sum^n_{j=1}\exp{\left(\beta_j - g_j\right)}},
	\end{equation}
	where $g_j$ is the estimated gradients to item $\beta_j$.
\end{lemma}
\begin{proof}
Since the learning for each opponent $-i$ follows the same rule, we mitigate the index $-i$ here.
\begin{align*}
    \sigma_{\beta'}(j) &= \frac{\exp{(\beta_j - g_j)}}{\sum^n_{i=1}\exp{(\beta_i - g_i)}} = \exp{\beta_j}\frac{\exp{- g_j}}{\sum^n_{i=1}\exp{\beta_i}\cdot\exp{ -g_i}} \\
    &= \frac{\exp{\beta_j}}{\sum^n_{i=1}\exp{\beta_i}}\cdot\frac{\exp{-g_j}}{\sum^n_{i=1}\frac{\exp{\beta_i}}{\sum^n_{k=1}\exp{\beta_k}}\exp{-g_i}} \\
    &= \sigma_{\beta}(j) \cdot \frac{\exp{-g_j}}{\sum^n_{i=1}\sigma_{\beta}(i)\exp{-g_i}}
\end{align*}
where $g_j=\eta\nabla_{\beta_j}u_{-i}(\pi_i, \sigma_{-i,\beta})$, $\eta$ is the learning rate.
\end{proof}

\addcontentsline{toc}{subsection}{A Proof of Theorem~\ref{theorem:warm_bound}}
\subsection*{A Proof of Theorem~\ref{theorem:warm_bound}}

As introduced in the main content, we consider a warm-start technique to reduce the re-training cost of meta-strategies and save training steps. To find a feasible meta-strategy $\bar{\sigma}^{'}_{-i}$ for warm-starting, we need to ensure that $\bar{\sigma}^{'}_{-i}$ follows two conditions proposed by Noam Brown et al. (2016)~\cite{brown2016strategy}, i.e. (1) cannot violate the bound of regret of the \textbf{latest epoch} $e+1$ and (2) cannot violate the Nash equilibrium of the \textbf{latest epoch} $e+1$. In our paper, we choose the average strategies $\bar{\sigma}^{'}_{-i}$ that satisfies $u^{e+1}_i(\bar{\pi}_i, \bar{\sigma}^{'}_{-i})=u^e_i(\bar{\pi}_i, \bar{\sigma}_{-i})$. Note that the utility function of epoch $e+1$ is different from the one of epoch $e$ since the dimension of opponent strategy is changed with a new introduced best response $\pi^{e+1}_{-i}$, i.e. $u^e_{i}(\bar{\pi}_i,\cdot) \in \mathbb{R}^k$ while $u^{e+1}_{i}(\bar{\pi}_i,\cdot) \in \mathbb{R}^{k+1}$.

\begin{theorem}\label{theorem:warm_bound_b}
     Suppose a substitute policy of $\bar{\sigma}_{-i} \in \Delta^e_{\Pi^r_{-i}}$ is $\bar{\sigma}^{'}_{-i} \in \Delta^{e+1}_{\Pi^r_{-i}}$, and it satisfies $u^{e+1}_i(\bar{\pi}_i, \bar{\sigma}^{'}_{-i})=u^e_i(\bar{\pi}_i, \bar{\sigma}_{-i})$, we have
     \begin{equation}\label{eq:warm_start_equation}
         \max_{\sigma_{-i}\in\Delta^{e+1}_{\Pi^r_{-i}}}\sum^T_{t=1}\left( u^{e+1}_{-i}(\pi^t_i, \sigma_{-i}) - u^{e+1}_{-i}(\pi^t_i, \bar{\sigma}^{'}_{-i}) \right) \le \epsilon_{-i},
     \end{equation}
     where $\epsilon_{-i}$ is the regret bound of epoch $e+1$, $\pi^t_i \in \Delta(\Pi_i)$.
\end{theorem}
\begin{proof}
	Noam Brown et al. have proved in~\cite{brown2016strategy} that (1) \emph{the growth of substitute regret has the same bound as the growth rate of normal regret (\cite[Lemma.~2]{brown2016strategy})} and (2) \emph{we can use a substitute strategy regret to prove convergence to a Nash Equilibrium just as we could use normal regret (\cite[Lemma.~3]{brown2016strategy})}.
	Thus, if we use a substitute strategy $(\bar{\pi}_i,\bar{\sigma}^{'}_{-i})$ that warm-start to $T$ iterations as Eq.~\ref{eq:warm_start_equation}, and play more $T' \ge 0$ iterations according to Algorithm~\ref{alg:deterministic}, then we have:
	\begin{equation}\label{eq:substitute_regret}
		\max_{\sigma_{-i}\in\Delta^{e+1}_{\Pi^r_{-i}}} \left( T\left(u^{e+1}_{-i}(\bar{\pi}_i, \sigma_{-i}) - u^{e+1}_{-i}(\bar{\pi}_i,\bar{\sigma}^{'}_{-i})\right) + \sum^{T'}_{t'=1}\left(u^{e+1}_{-i}(\pi^{t'}_i,\sigma_{-i})-u^{e+1}_{-i}(\pi^{t'}_i,\sigma^{t'}_{-i})\right)\right) \le \epsilon_{-i}.
	\end{equation}
    As the summation of regret from $T+1$ to $T+T'$ is non-negative, then we know Eq.~\ref{eq:warm_start_equation} still holds.
    Since $\max_{\sigma_{-i}}u^{e+1}_{-i}(\bar{\pi}_i,\sigma_{-i}) \ge u^{e+1}_{-i}(\bar{\pi}_i,\bar{\sigma}^{'}_{-i})$ and $\max_{\pi_{i}}u^{e+1}_{i}(\pi_i, \bar{\sigma}_{-i}) \ge u^{e+1}_{i}(\bar{\pi}_i, \bar{\sigma}_{-i})$, so the substitute strategies $(\bar{\pi}_i,\bar{\sigma}^{'}_{-i})$ satisfies the convergence of $(\epsilon_i,\epsilon_{-i})$-Nash equilibrium of epoch $e+1$ (See Theorem~\ref{theorem:averaging} or~\ref{theorem:noam_theorem_2}). Then $(\bar{\pi}_i,\bar{\sigma}^{'}_{-i})$ are feasible warm-start strategies.
\end{proof}
\addcontentsline{toc}{subsection}{A Proof of Lemma~\ref{lemma:beta_opt}}
\subsection*{A Proof of Lemma~\ref{lemma:beta_opt}}

As for the exact bound of regret, we will give a proof in Theorem~\ref{theorem:warm_bound_b}. Note that the utility function of epoch $e+1$ is different from the one of epoch $e$ since the dimension of opponent strategy is changed with a new introduced best response $\pi^{e+1}_{-i}$, i.e. $u^e_{-i}(\bar{\pi}_i,\cdot) \in \mathbb{R}^k$ while $u^{e+1}_{-i}(\bar{\pi}_i,\cdot) \in \mathbb{R}^{k+1}$. Thus, before the optimization of $\beta$, we need to slightly run a simulation to estimate the item $u_{-i}(\bar{\pi}_i,\pi^{e+1}_{-i})$ in $u^{e+1}_{-i}$.

\begin{lemma}\label{lemma:beta_opt_b}
	Let $k=|\Pi^{r,e+1}_{-i}|$, $\bar{\sigma}^{'}_{-i}$ be parameterized by $\beta_{-i}=\left[\beta_{-i,1},\beta_{-i,2},\dots,\beta_{-i,k}\right]$, $\bar{\sigma}^{'}_{-i}(k)$ the $k$-th item of $\bar{\sigma}^{'}_{-i}$, $x=[\bar{\sigma}^{'}_{-i}(1),\dots,\bar{\sigma}^{'}_{-i}(k - 1)]^{\top}$, $\bar{l}^e_{-i}$ is $-i$'s average loss vector to $\bar{\pi}_i$ at epoch $e$. Then a feasible initial of $\beta^{e+1}_{-i}$ could satisfy
	\begin{equation}\label{eq:beta_opt}
		\beta^{e+1}_{-i} = \arg\min_{\beta_{-i}}\parallel (x-\bar{\sigma}_{-i})^{\top} \bar{l}^e_{-i} - \bar{\sigma}^{'}_{-i}(k) u_{-i}(\bar{\pi}_i,\pi^{e+1}_{-i})\parallel_2.
	\end{equation}
\end{lemma}
\begin{proof}
	Considering the error between $\xi = u^e_{-i}(\bar{\pi}_i,\bar{\sigma}_{-i})-u^{e+1}_{-i}(\bar{\pi}_i,\bar{\sigma}^{'}_{-i})$. It can be rewritten as
	\begin{align*}
		\xi &= -\bar{\sigma}^{\top}_{-i}\bar{l}^e_{-i} - \left(-x^{\top}\bar{l}^e_{-i} + \bar{\sigma}^{'}_{-i}(k)u_{-i}(\bar{\pi}_i,\pi^{e+1}_{-i}) \right)\\
		&=(x-\bar{\sigma}_{-i})^{\top}\bar{l}^e_{-i} - \bar{\sigma}^{'}_{-i}(k)u_{-i}(\bar{\pi}_i,\pi^{e+1}_{-i}).
	\end{align*}
    Thus, we can optimize $\beta_{-i}$ by minimizing the error $\xi$ as $\beta^{e+1}_{-i}=\arg\min_{\beta_{-i}}\parallel \xi \parallel_2$
\end{proof}

We compute the initial of $\sigma^{e+1}_{-i}$ as Lemma~\ref{lemma:beta_opt_b}. Note that there is another solution as $\bar{\sigma}^{'}_{-i}(k)=0, \bar{\sigma}_{-i}=x$. In that case, the newly introduced policy $\pi^{e+1}_{-i}$ will have no chance to be sampled, causing biased or even wrong solution. Thus, we consider adding a regularization to Eq.~\ref{eq:beta_opt} that maximizes the entropy of $\bar{\sigma}^{'}_{-i}$, i.e. $H(\bar{\sigma}^{'}_{-i}) = -\bar{\sigma}^{'\top}_{-i}\log{\bar{\sigma}^{'}_{-i}}$ and $\beta^{e+1}_{-i}=\arg\min_{\beta_{-i}}\parallel \xi \parallel_2 - \lambda H(\bar{\sigma}^{'}_{-i})$, where $\lambda > 0$.

\addcontentsline{toc}{subsection}{A Proof of Theorem~\ref{theorem:regret_meta_warm}}
\subsection*{A Proof of Theorem~\ref{theorem:regret_meta_warm}}

\begin{theorem}[\emph{Regret Bound of \framework{}}]\label{theorem:regret_meta_warm_b}
Let $l_1,l_2,\dots,l_T$ be a sequence of loss vectors player by an adversary, and $\langle \cdot,\cdot \rangle$ be the dot product, then Algorithm~\ref{alg:deterministic} is a no-regret algorithm with
	\begin{equation*}
		\frac{1}{T}\left(\sum^T_{t=1}\langle\sigma_t,l_t\rangle-\min_{\sigma \in \Delta(\Pi^r_t)}\sum^T_{t=1}\langle\sigma,l_t\rangle\right)\le \frac{\sqrt{\log{[(k+1)k/2]}}}{\sqrt{2T}},\text{ where }k\text{ is the size of }\Pi^r.
	\end{equation*}
\end{theorem}
\begin{proof}
	Note that the loss vector is the opposite of utility at each iteration, i.e., $l_t=-u_t$. The proof follows the approach of Le Cong Dinh et al. (2021)~\cite{le2021online}. Denote $|T_1|, \dots, |T_k|$ be the learning horizon of each epoch. During each $|T_j|$, the restricted policy set is denoted as unchanged $\Pi^{r,i}$. In the case of finite $k$, we have:
	\begin{equation}
			\sum^k_{i=1}|T_i|=T.
		\end{equation}
	For each training epoch $i$, following the regret bound of MWU, we have:
	\begin{equation}
			\sum^{|T^{i+1}|}_{t=|T^i|+1} \langle \sigma_t, l_t \rangle - \min_{\sigma\in\Delta(\Pi^r_i)}\sum^{|T^{i+1}|}_{t=|T^i|+1}\langle \sigma, l_t \rangle \le \sqrt{\frac{|T_i|\log{|\Pi^{r,i}|}}{2}}, \text{ where }|T^i|=\sum^{i}_{j=1}|T_{j}|.
		\end{equation}
	Considering all time horizon, for $i=1,\dots,k$, we have:
	\begin{align}
			\sum^k_{i=1}\sqrt{\frac{|T_i|\log{|\Pi^{r,i}|}}{2}} &\ge \sum^T_{t=1}\langle\sigma_t,l_t\rangle - \sum^k_{i=1}\min_{\sigma\in\Delta(\Pi^{r,i})}\sum^{|T^{i+1}|}_{t=|T^i|+1}\langle \sigma, l_t \rangle \\
			&\ge \sum^T_{t=1}\langle \sigma_t, l_t \rangle - \min_{\sigma\in\Delta(\Pi^r)}\sum^k_{i=1}\sum^{|T^{i+1}|}_{t=|T^i|+1}\langle \sigma, l_t \rangle \\
			&=\sum^T_{t=1}\langle \sigma_t, l_t \rangle - \min_{\sigma \in \Delta(\Pi^r)} \sum^T_{t=1}\langle \sigma, l_t \rangle.
		\end{align}
	Note that $i=|\Pi^{r,i}|$, then $\sum^k_{i=1}\sqrt{|T_i|\log{|\Pi^{r,i}|}} \le \sum^k_{i=1}\sqrt{T\log(i)}$. As the sum of logarithms is concave, so we have $\sum^k_{i=1}\sqrt{T\log{(i)}} \le \sqrt{T\log{\left(\sum^k_{i=1}i\right)}}=\sqrt{T\log{[(k+1)k/2]}} $
	then we have 
	\begin{equation*}
		\frac{\sqrt{\log{[(k+1)k/2]}}}{\sqrt{2T}} \ge \frac{1}{T}\left( \sum^T_{t=1}\langle \sigma_t, l_t \rangle - \min_{\sigma \in \Delta(\Pi^r)} \sum^T_{t=1}\langle\sigma,l_t\rangle \right).
	\end{equation*}
\end{proof}

In addition, we provide a comparison of regret bound on existing solvers in Appendix~\ref{sec:comparison_regret}.

\addcontentsline{toc}{subsection}{A Proof of Theorem~\ref{theorem:convergence_rate_warm}}
\subsection*{A Proof of Theorem~\ref{theorem:convergence_rate_warm}}

Before the proof of convergence rate of \framework{}, we want to clarify that the reinforcement learning procedure for the best response $\pi_i$ should be an equivalence of MWU.
Fortunately, many reinforcement learning algorithms can be treated like that as their policy function is modeled as a Boltzman distribution~\cite{cen2022independent}. Thus, \framework{} can implement such a procedure with algorithms like Soft Actor Critic~\cite{haarnoja2018soft}, Soft Q-learning~\cite{haarnoja2017reinforcement}, DQN with Boltzman sampling~\cite{yang2018mean} and others.
With this condition, we can apply the Theorem~\ref{theorem:warm_bound_b} to $\pi_i$ and then derive the convergence rate as follows.

\begin{theorem}[\emph{Convergence Rate of \framework{}}]\label{theorem:convergence_rate_warm_b}
     Let $k$, $N$ denote the size of restricted policy sets $\Pi^r_{-i}$ and $\Pi_i$. Then the learning of Algorithm~\ref{alg:solve_urr} will converge to the Nash equilibrium with the rate:
     \begin{equation*}
     	\epsilon_T=\sqrt{\frac{\log{[(k+1)k/2]}}{2T}} + \sqrt{\frac{\log{[(N+1)N/2]}}{2T}}.
     \end{equation*}
\end{theorem}
\begin{proof}
Using the regret bound of Theorem~\ref{theorem:regret_meta_warm_b} we have:
\begin{align*}
    \max_{\pi_i}\frac{1}{T}\left( \sum^T_{t=1}u_i(\pi_i,\sigma^t_{-i}) - u_i(\pi^t_i,\sigma^t_{-i})\right) &\le \epsilon_i,\\
    \max_{\sigma_{-i}}\frac{1}{T}\left( \sum^T_{t=1}u_{-i}(\pi^t_i,\sigma_{-i}) - u_{-i}(\pi^t_i,\sigma^t_{-i})\right) &\le \epsilon_{-i},
\end{align*}
where $\epsilon_i=\sqrt{\frac{\log{[(N+1)N/2]}}{2T}}$ and $\epsilon_{-i}=\sqrt{\frac{\log{[(k+1)k/2]}}{2T}}$. From the above inequalities and $u_{-i}(\cdot,\cdot)=-u_i(\cdot,\cdot)$ in zero-sum games, we can derive that
\begin{align*}
    u_i(\bar{\pi}_i, \bar{\sigma}_i) &\ge \min_{\sigma_{-i} \in \Delta(\Pi^r_{-i})} u_i(\bar{\pi}_i,\sigma_{-i}) \ge \frac{1}{T} \sum^T_{t=1}u_i(\pi_t,\sigma_t) - \epsilon_{-i}\\
    &\ge \max_{\pi_i}u_i(\pi_i,\bar{\sigma}_{-i}) - \epsilon_i - \epsilon_{-i}.
\end{align*}
By symmetry,  we have
\begin{align*}
	u_i(\bar{\pi}_i, \bar{\sigma}_i) &\le \max_{\pi_i}u_i(\pi_i,\bar{\sigma}_{-i}) \le \frac{1}{T}\sum^T_{t=1}u_i(\pi^t_i,\sigma^t_{-i}) + \epsilon_i\\
	&\le \min_{\sigma_{-i}}u_i(\bar{\pi}_i,\sigma_{-i}) + \epsilon_{-i} + \epsilon_i.
\end{align*}
Thus, with $\epsilon_T=\epsilon_i+\epsilon_{-i}$, we have
\begin{align*}
	\max_{\pi_i}u_i(\pi_i,\bar{\sigma}_{-i}) - \epsilon_T \le u_i(\bar{\pi}_i,\bar{\sigma}_{-i}) \le \min_{\sigma_{-i}}u_i(\bar{\pi}_i,\sigma_{-i}) + \epsilon_T.
\end{align*}
By definition, we have $(\bar{\pi}_i, \bar{\sigma}_{-i})$ is $\epsilon_T$-Nash equilibrium.
\end{proof}

\section{Implementation and Hyper-parameter Selection}\label{appendix:parameters}

\paragraph{Learning Meta-strategies and Best-response.} We introduce the algorithm for learning meta-strategies in Algorithm~\ref{alg:deterministic}.
We train best response policies with off-policy reinforcement learning algorithm, DQN. For all involved PSRO baselines in our paper, the selected implementation for best response learning is DQN.

\paragraph{Parameter Selection.}
We keep the consistency on the implementation of policy support in each PSRO-based method in this paper. Expressly, the network is set to 4 Dense layers, 256 units each. The learning rate for reinforcement learning policy is set to 0.01. For the hyper-parameters of meta-strategy, the window size $L$ is set to 100 episodes, the learning rate is 0.01 for Kuhn Poker, and 0.005 for other environments. The number of parallel workers is set to 4 for \framework{} and P-PSRO in all experiments.

\section{Baselines}

\paragraph{Self-Play.}
Self-play is an open-ended learning algorithm for multi-agent reinforcement learning~\cite{hernandez2019generalized}. In the training process, self-play generates a sequence of policies and keeps training policies against the newest opponents. This algorithm outperforms in some classic games, such as Go and Chess. However, self-play fails in nontransitive games.

\paragraph{ Policy Space
Response Oracles (PSRO).}
PSRO algorithm is well described above in previous sections. It provides
an iterative solution to solve the approximation of Nash equilibrium for large games~\cite{lanctot2017unified}. PSRO iteratively trains new policies against a meta-strategy of opponent population and expends policy populations with the current well-trained policy. 

\paragraph{Rectified PSRO (PSRO-rN).}
PSRO-rN is a variant of PSRO that aims to solve non-transitive zero-sum games~\cite{balduzzi2019open}, such as rock-paper-scissors. This algorithm involves rectified Nash response to construct adaptive sequences of objectives for non-transitive games. Policies in PSRO-rN only train against others that they
already beat.

\paragraph{Mixed Oracles.}
Mixed Oracles is another variant of PSRO that aims to improve computational efficiency by reducing training costs~\cite{smith2020iterative}. At each iteration, it utilizes knowledge of former iterations, thus only needing to train current policies against the newest opponent.

\paragraph{Pipeline-PSRO (P-PSRO).}
To further accelerate the training process of PSRO, P-PSRO is proposed to parallelize the training process~\cite{mcaleer2020pipeline}. Compared to other parallel algorithms, such as DHC, which fail to converge in some cases, P-PSRO maintains a parallel pipeline of learning workers with convergence guarantees.

\section{Environment Details}\label{appendix:env}
We introduce more details about the random symmetric games and mulit-agent gathering environments here.

\paragraph{Random Symmetric Games.}
McAleer et al.~\cite{mcaleer2020pipeline} introduce the games to investigate the performance of PSRO-based methods in high-dimensional symmetric games (SymGame). In this experiment, we generated random symmetric zero-sum matrices with different dimension $n$. For a given matrix, elements in the upper triangle are distributed uniformly: $\forall i < j \le n, a_{i,j} \sim \textsc{Uniform}(-1,1)$ and for the lower triangle, the elements are set to be the negative of its diagonal counterpart: $\forall j <i \le n, a_{i,j} = -a_{j,i}$. The diagonal elements are equal to zero: $a_i,i = 0$. The matrix defines the utility of two pure strategies to the row player. A strategy $\pi \in \Delta^n$ is a distribution over the $n$ pure strategies of the game given by the rows (or equivalently, columns) of the matrix.

\paragraph{Multi-agent Gathering.}
We introduce two multi-agent gathering environments in this paper, the \emph{Gathering Small} and the \emph{Gathering Open}. For each environment, the agent number is set to 2, and the difference between them is that the \emph{Gathering Open} has a much bigger map than \emph{Gathering Small}. Therefore, the agents need to explore a higher dimensional state space in the \emph{Gathering Open} than the smaller one.

\section{Additional Experimental Results}\label{appendix:results}

\subsection{Non-transitive Mixture Game}\label{appendix:results_non_trans}
We test six PSRO-based algorithms in the \emph{non-transitive mixture game} to investigate the exploration efficiency. We also introduce the results of \textsc{NashConv} to compare the performance. We find that \framework{} outperforms all other algorithms in this game, and performs the highest exploration efficiency.

\begin{figure}[h!]
	\centering
	\subfigure[PSRO]{\label{fig:mixture_psro}\includegraphics[width=.3\textwidth]{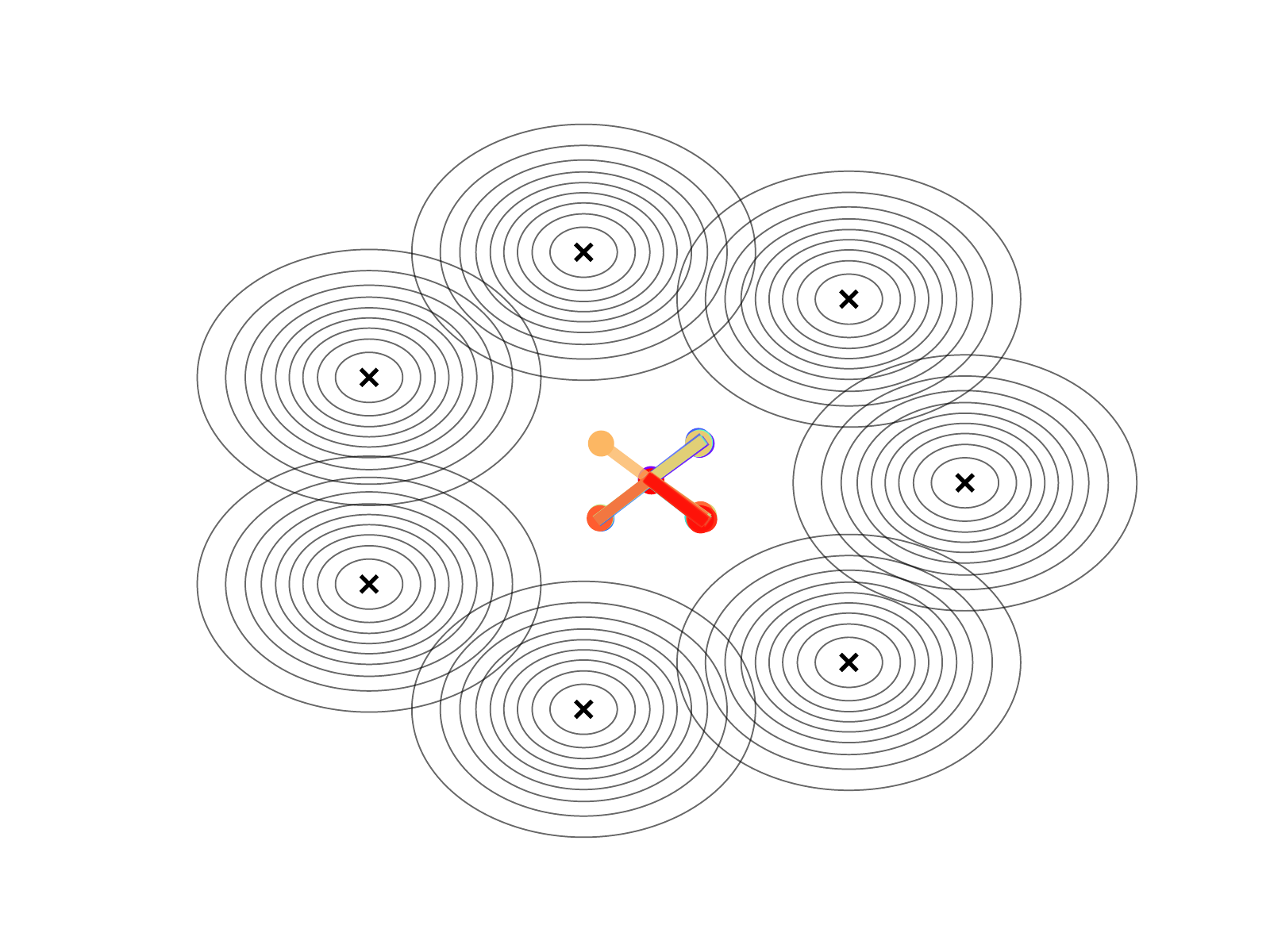}}
	\subfigure[Mixed-Oracles]{\label{fig:mixture_mixed_oracles}\includegraphics[width=.3\textwidth]{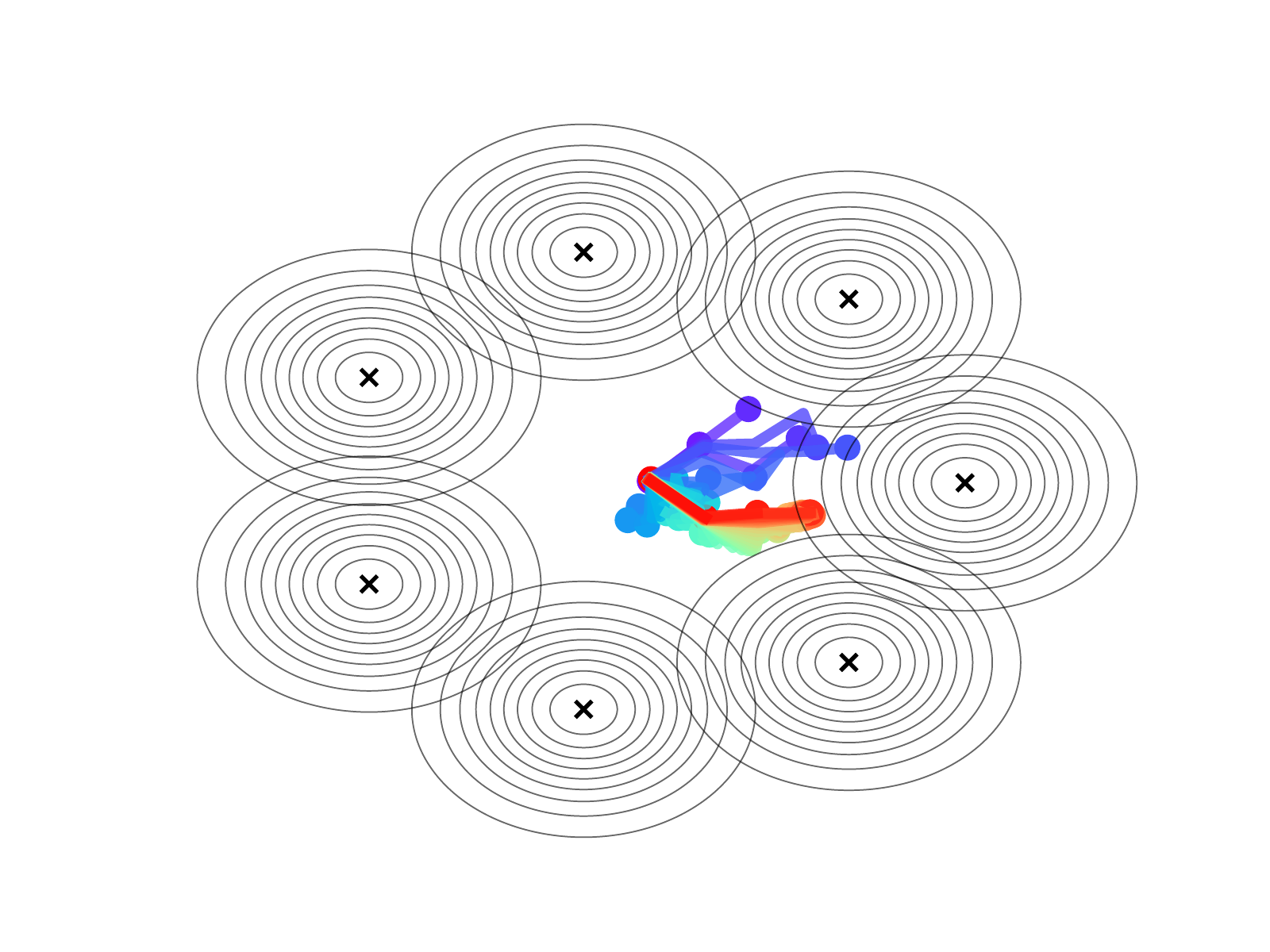}}
    \subfigure[PSRO-rN]{\label{fig:mixture_psro_rn}\includegraphics[width=.3\textwidth]{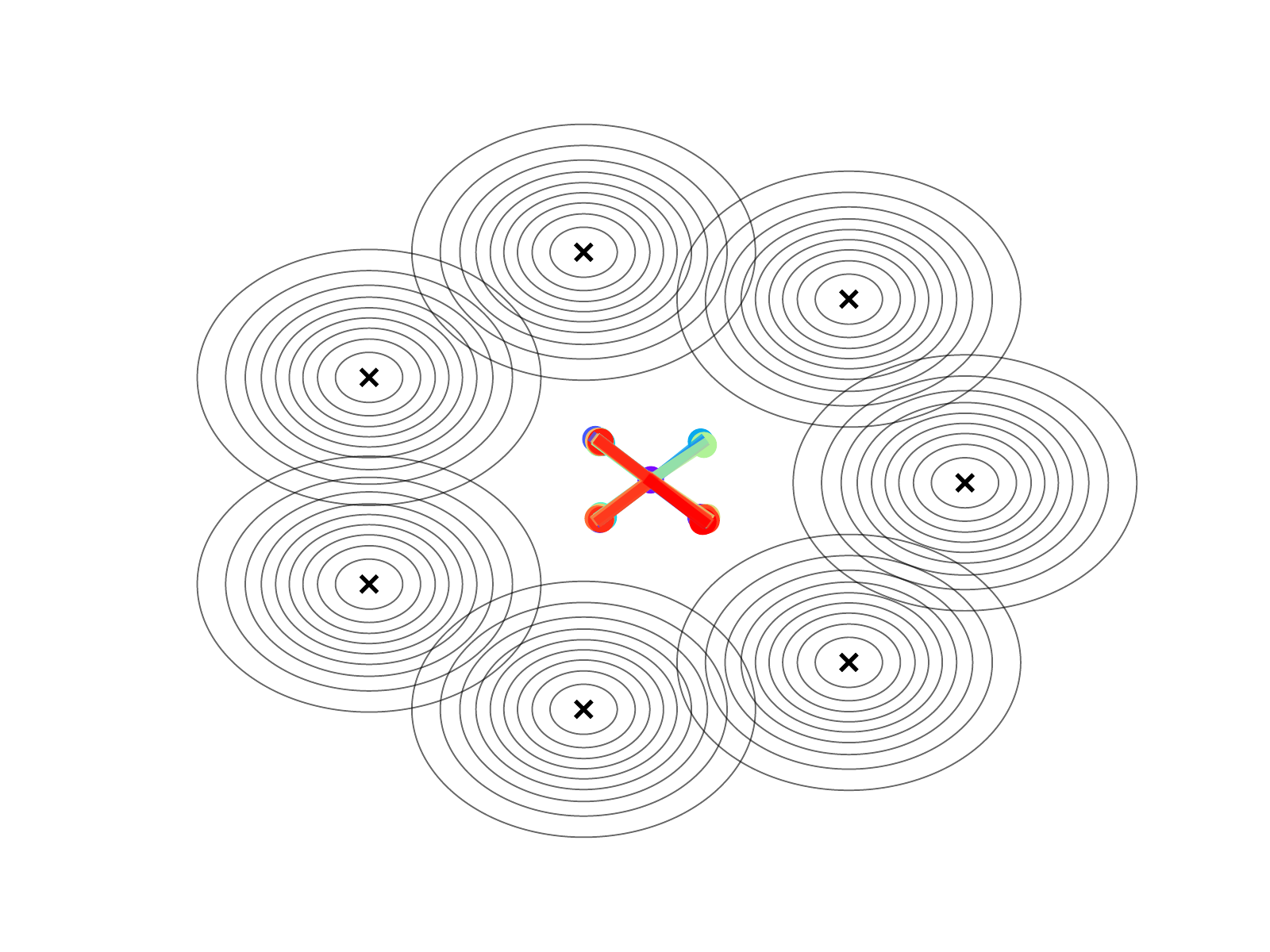}}
	\subfigure[P-PSRO]{\label{fig:mixture_p2sro}\includegraphics[width=.3\textwidth]{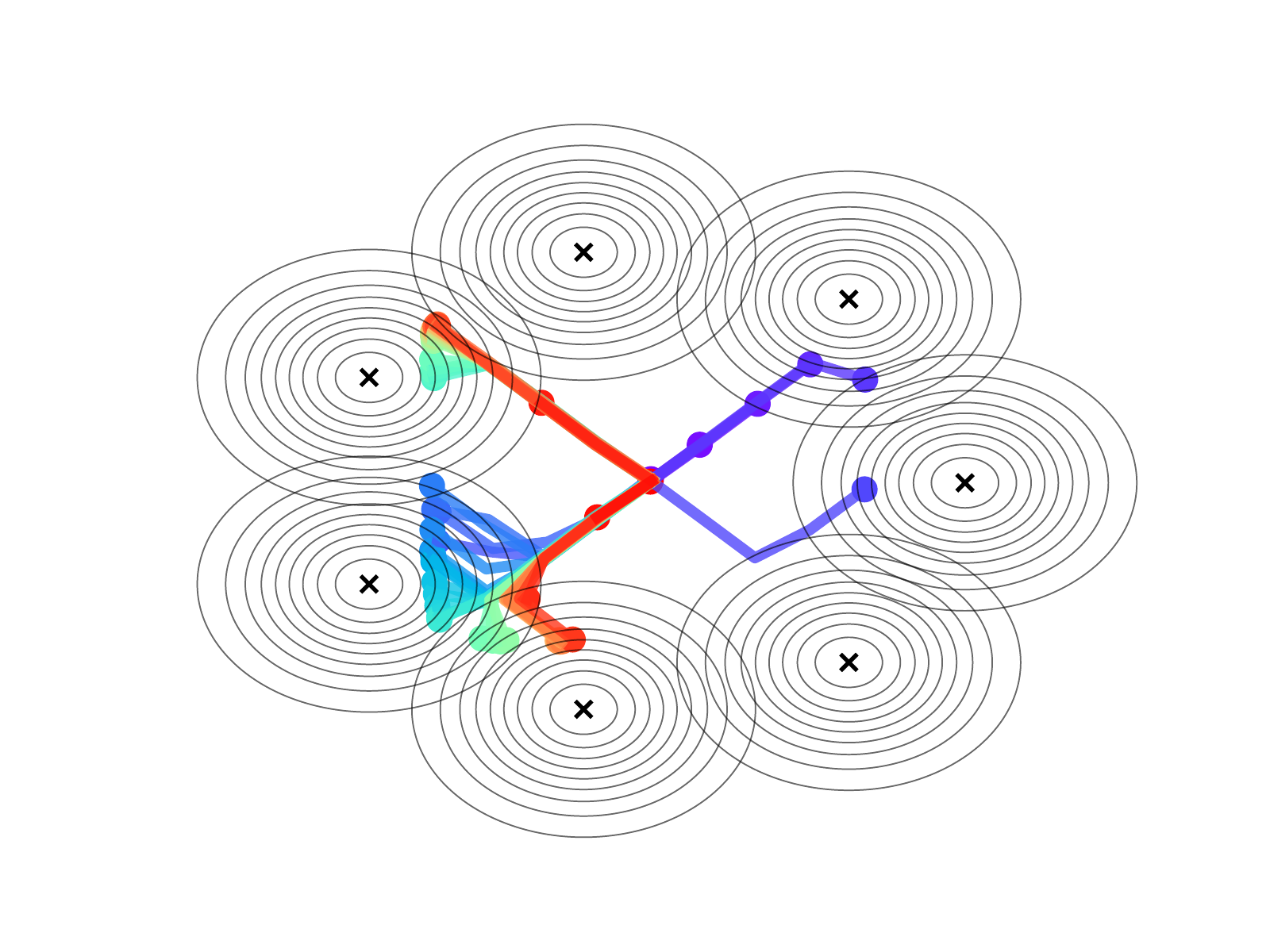}}
	\subfigure[NEPSRO]{\label{fig:mixture_nepsro}\includegraphics[width=.3\textwidth]{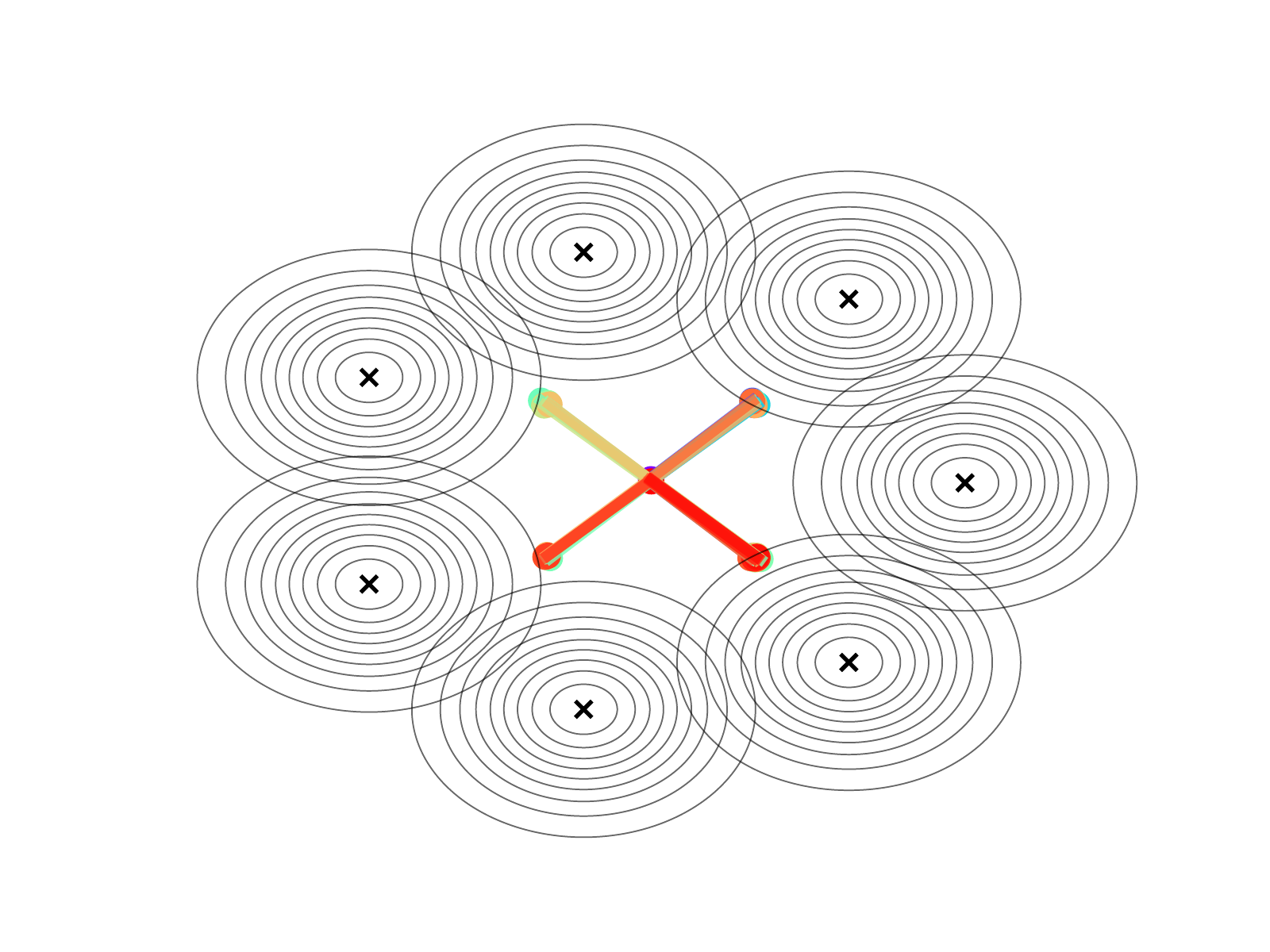}}
	\subfigure[EPSRO]{\label{fig:mixture_epsro}\includegraphics[width=.3\textwidth]{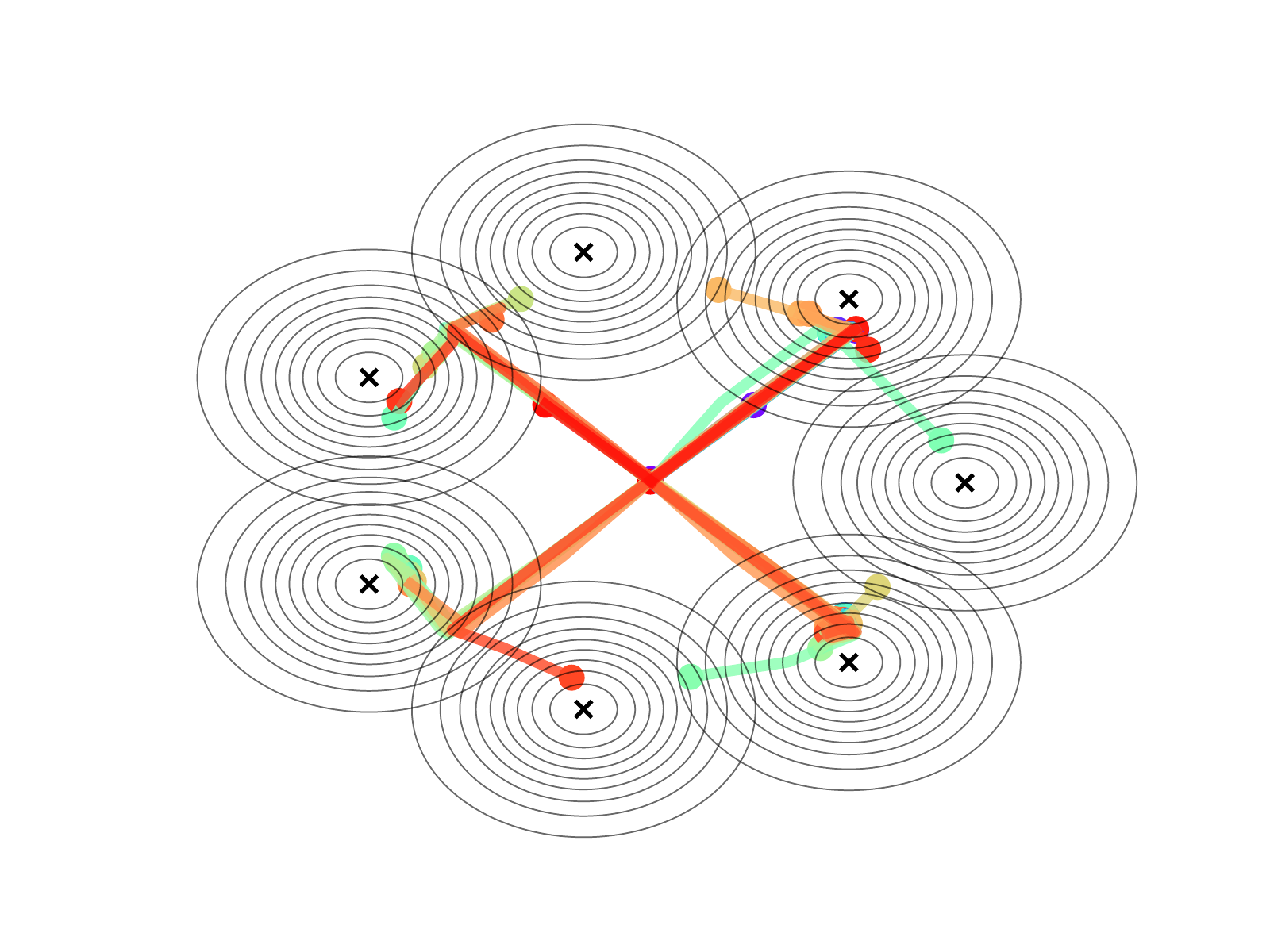}}
    \caption{Exploration trajectories on \textit{Non-transitive Mixture Games}. The more trajectories close to the centers of Gaussian, the higher the exploration efficiency of the algorithm. \framework{} outperform all selected baselines since it explored all centers.}
    \label{fig:non_mixture}
\end{figure}

\subsection{Random Symmetric Matrix Game}\label{appendix:results_random}
We compare the \textsc{NashConv} over iteration of \framework{} with PSRO, P-PSRO, Rectified PSRO, Self-Play, Mixed-Oracles and naive \framework{} without pipeline training (NEPSRO). We run 5 experiments for each set of dimension. The dimensions of size including 15, 30, 45, 60 and 120. The learning rates is set to 0.5, and 4 parallel threads for parallel algorithms. We find that \framework{} performs better than all other algorithms in every dimension setting.

\begin{figure}[h!]
	\centering
	\subfigure[Cardinality when dim=15]{\label{fig:random_15_card}\includegraphics[width=.24\textwidth]{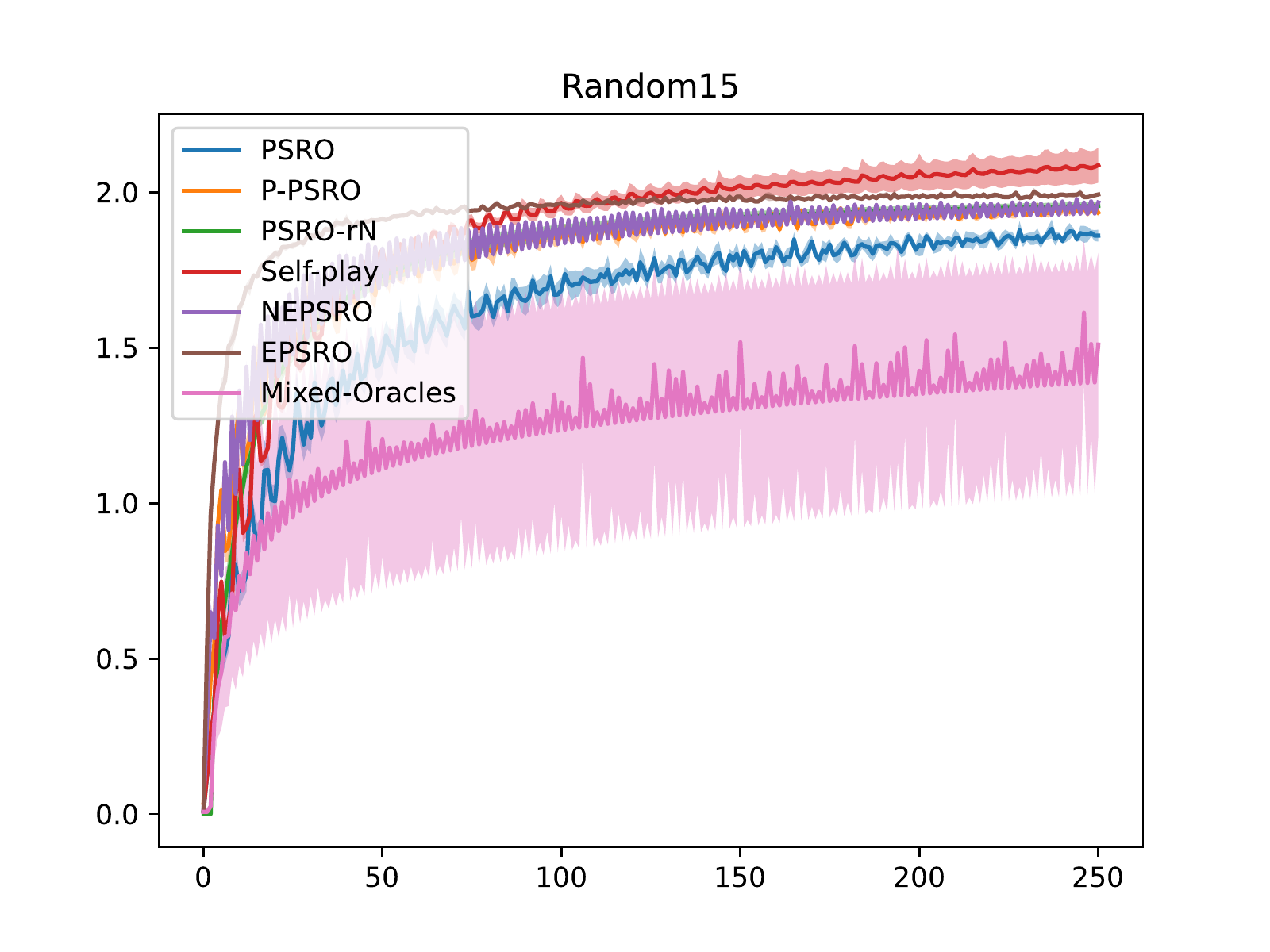}}
	\subfigure[\textsc{NashConv} when dim=15]{\label{fig:random_15_exp}\includegraphics[width=.24\textwidth]{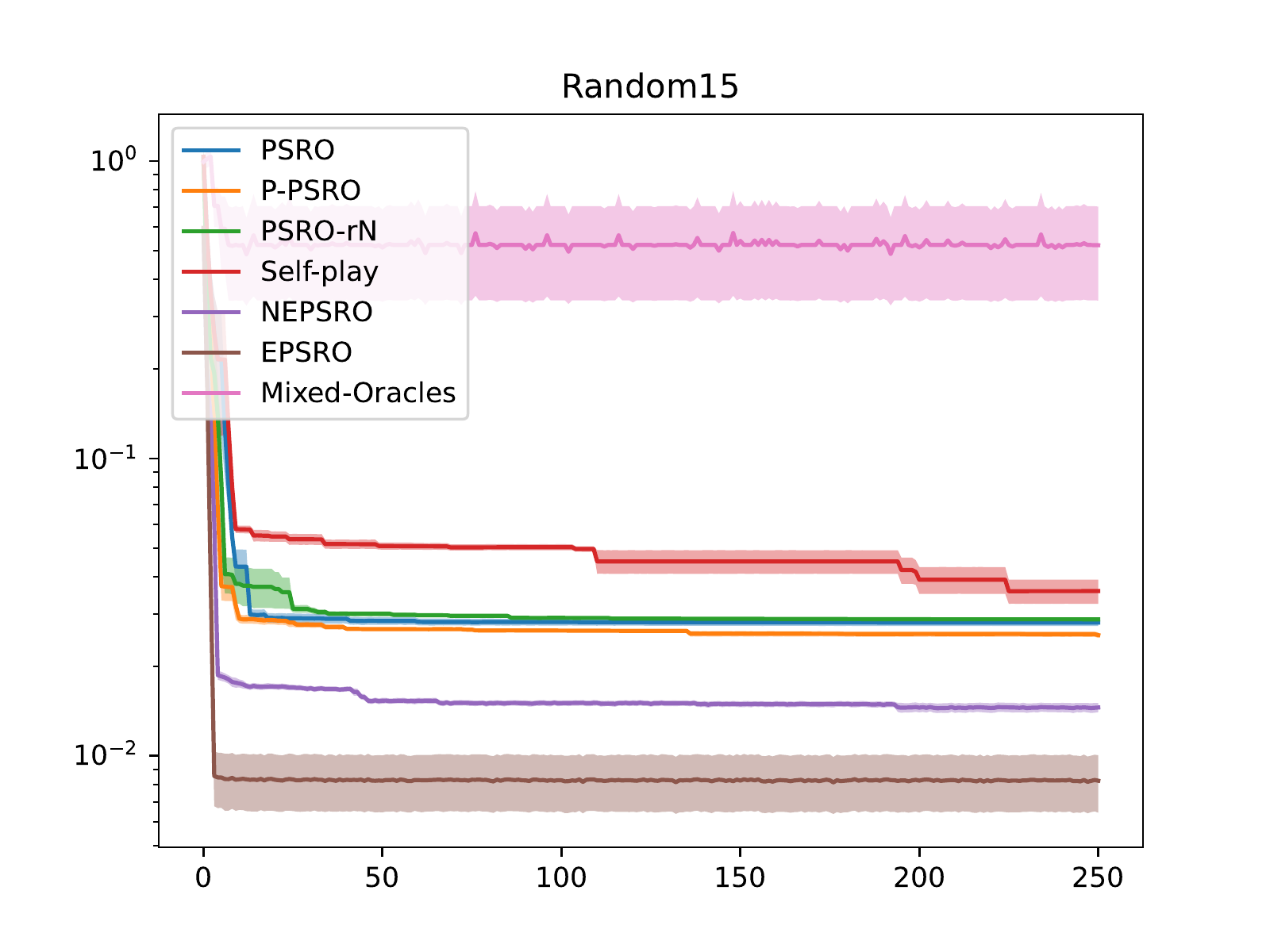}}
	\subfigure[Cardinality when dim=30]{\label{fig:random_30_card}\includegraphics[width=.24\textwidth]{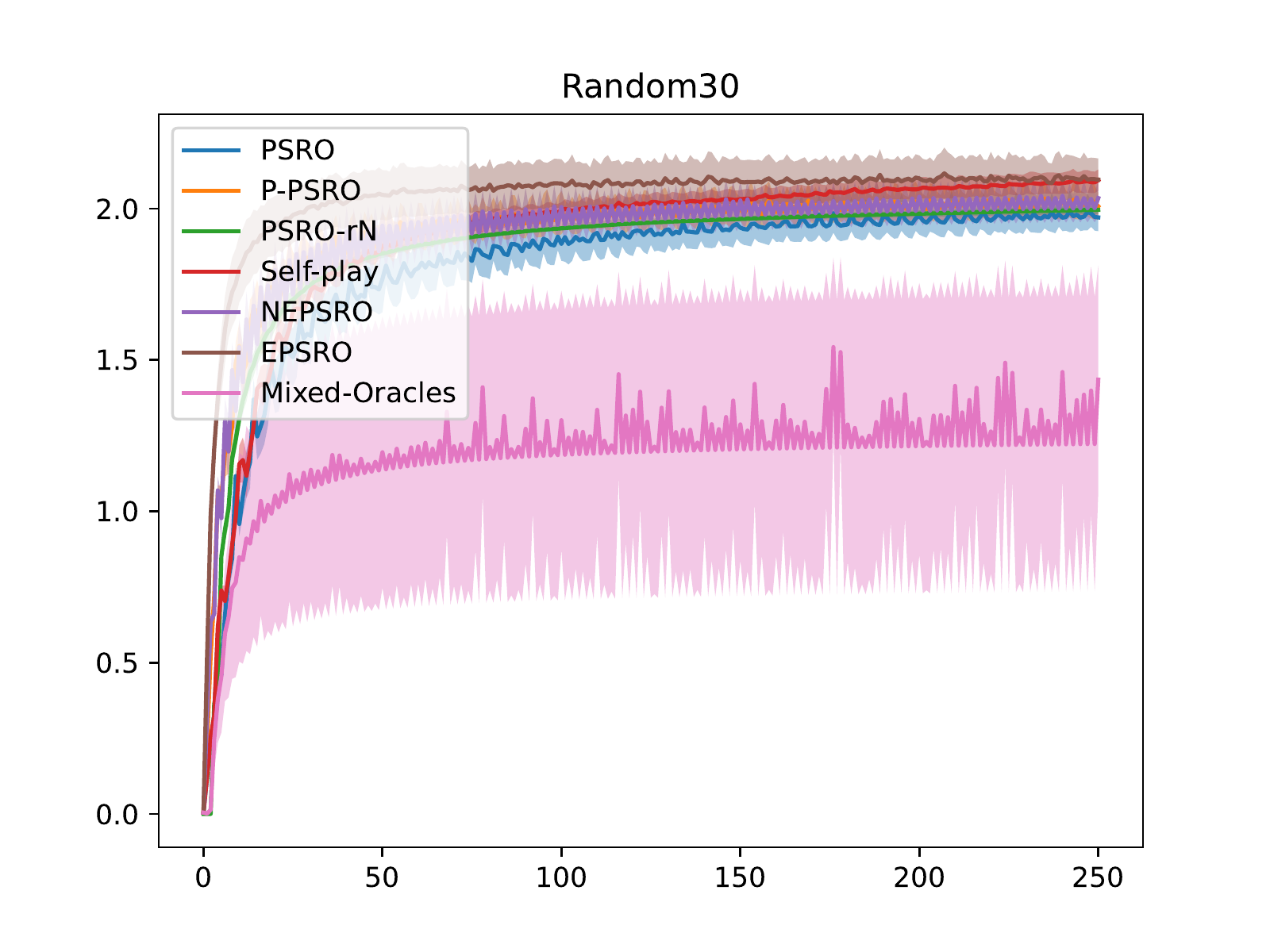}}
	\subfigure[\textsc{NashConv} when dim=30]{\label{fig:random_30_exp}\includegraphics[width=.24\textwidth]{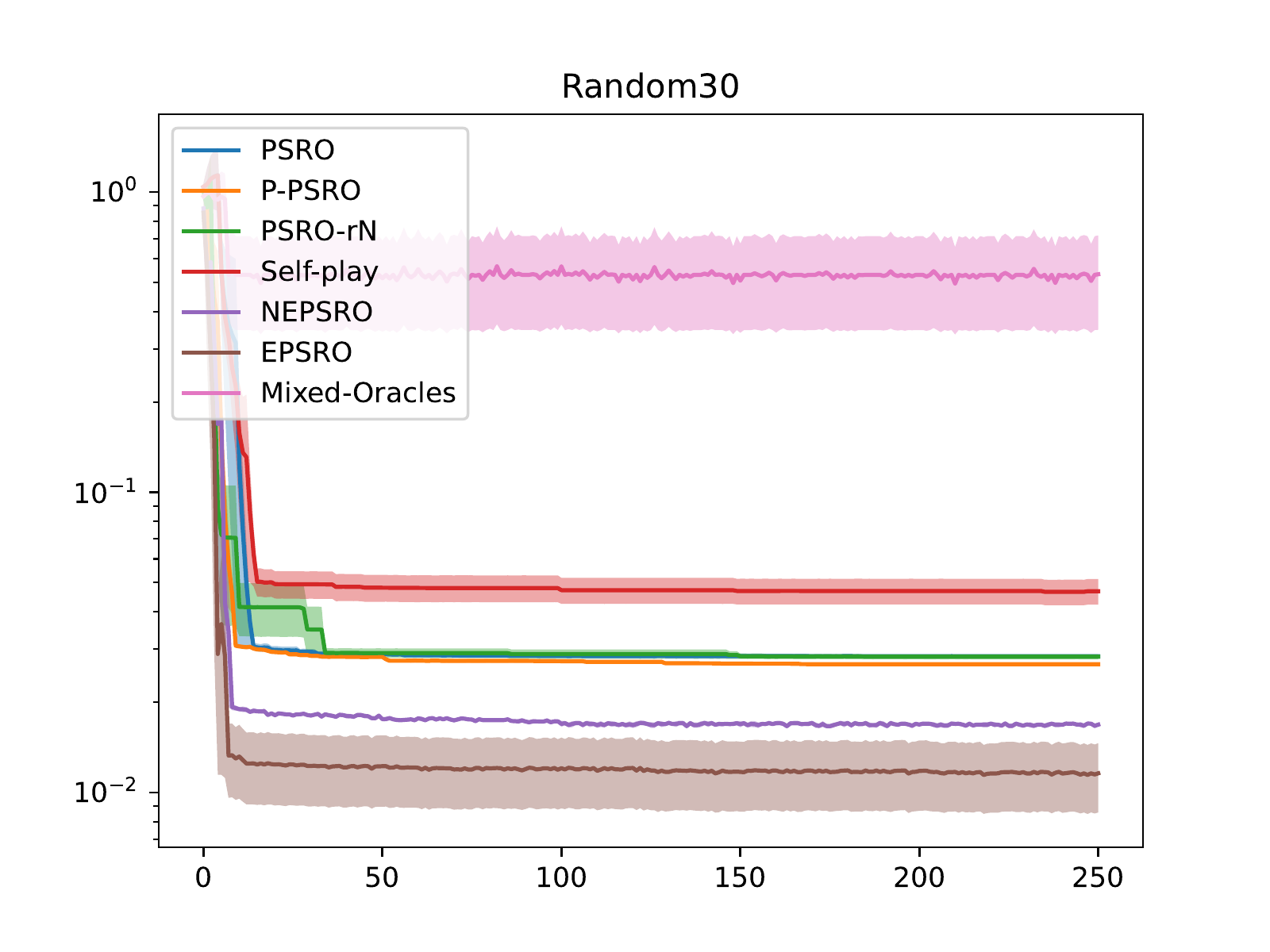}}
	\subfigure[Cardinality when dim=45]{\label{fig:random_45_card}\includegraphics[width=.24\textwidth]{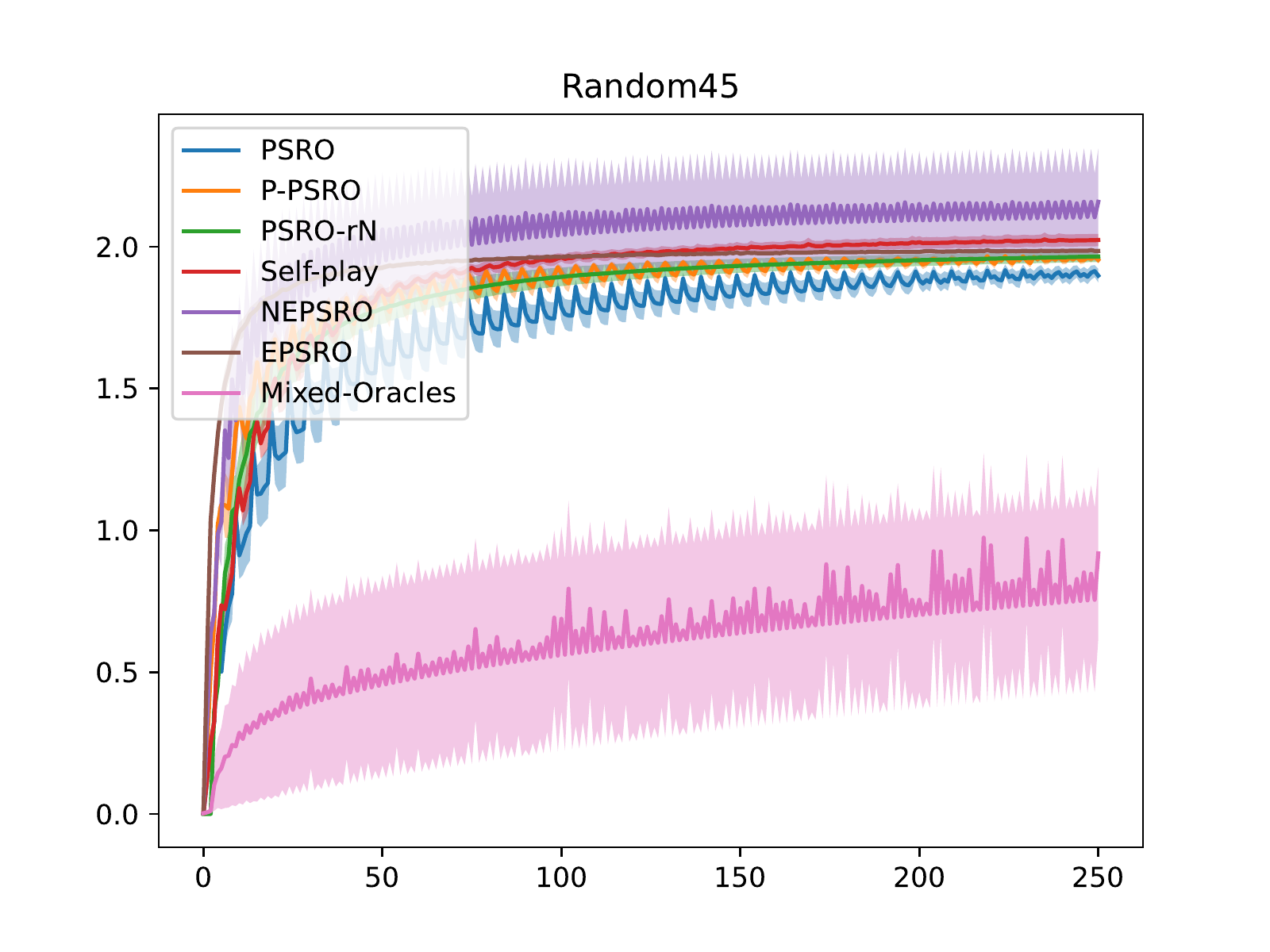}}
	\subfigure[\textsc{NashConv} when dim=45]{\label{fig:random_45_exp}\includegraphics[width=.24\textwidth]{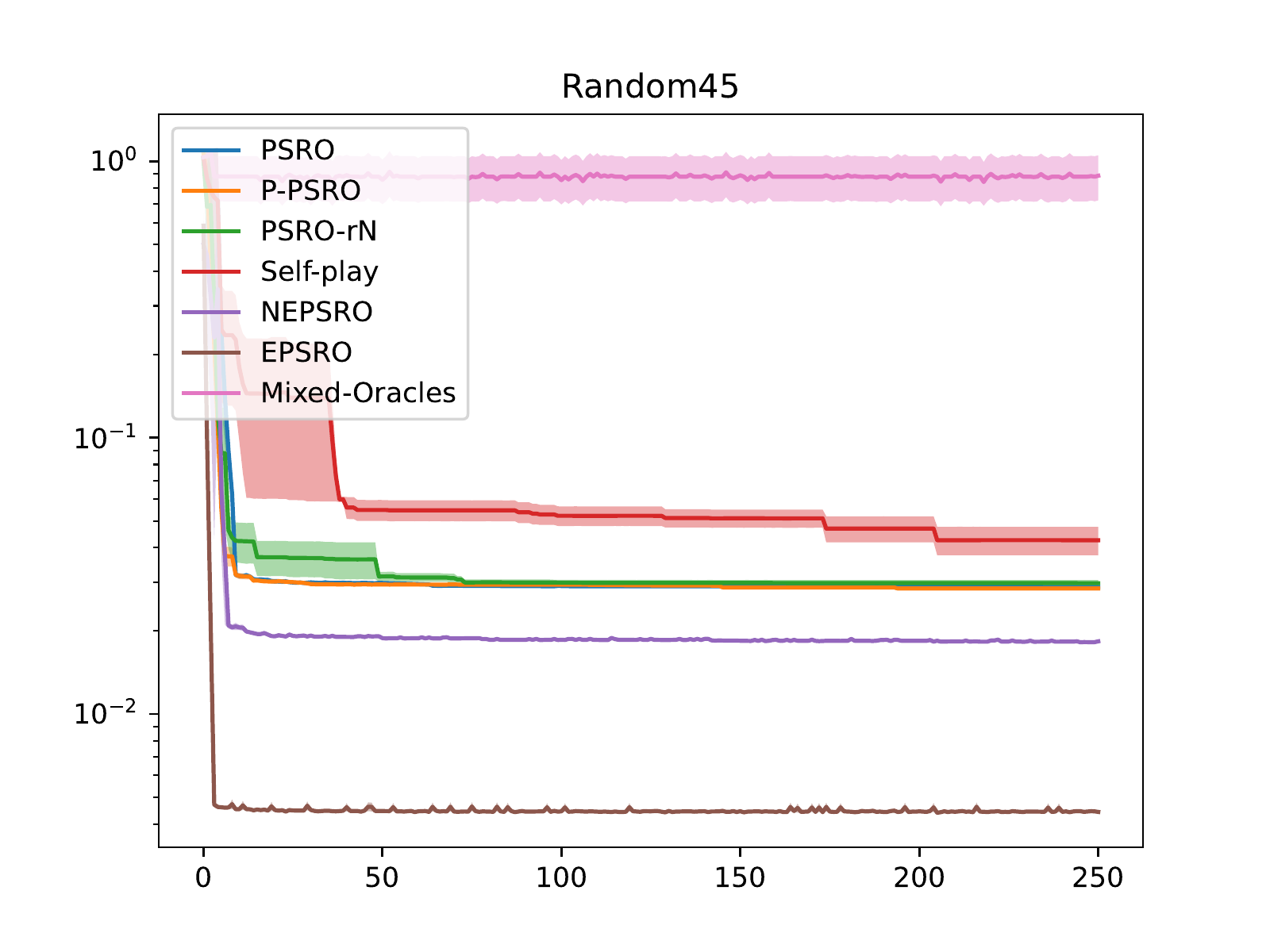}}
	\subfigure[Cardinality when dim=60]{\label{fig:random_60_card}\includegraphics[width=.24\textwidth]{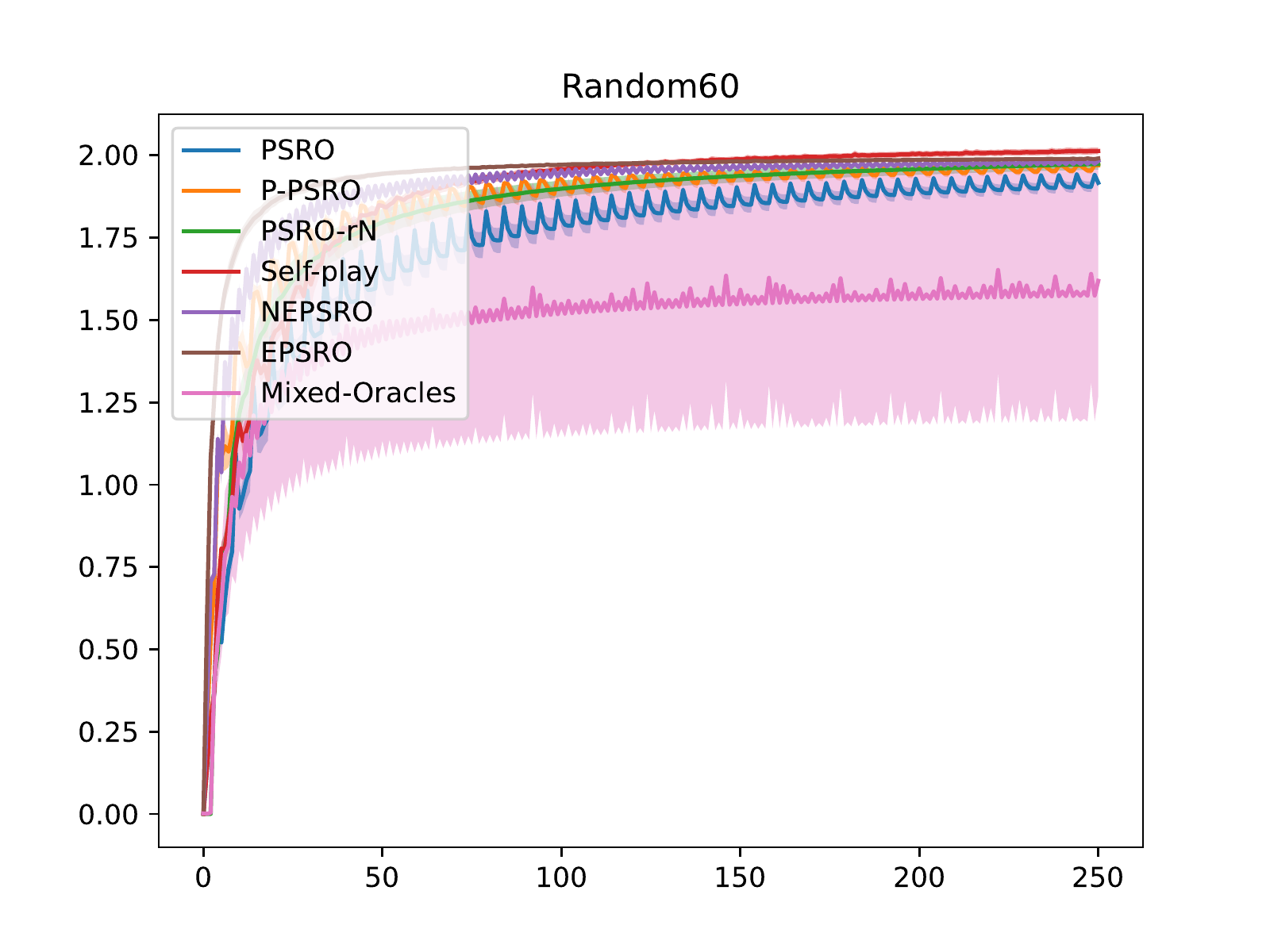}}
	\subfigure[\textsc{NashConv} when dim=60]{\label{fig:random_60_exp}\includegraphics[width=.24\textwidth]{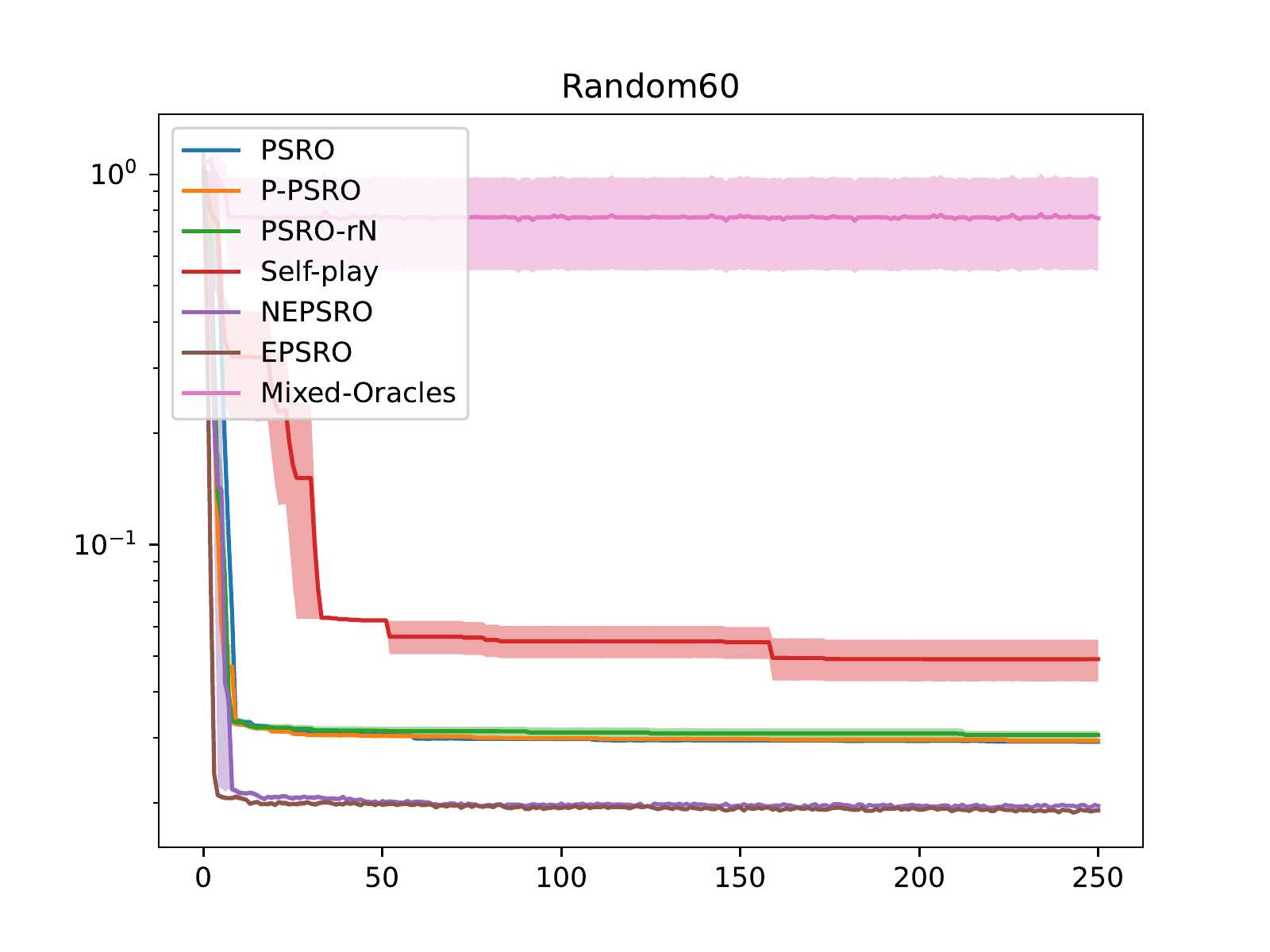}}
    \caption{Comparison on \textsc{NashConv} and cardinality for random symmetric matrix games with different dimensions.}
    \label{fig:random_game}
\end{figure}

\subsection{Multi-agent Gathering}\label{appendix:results_gathering}
Multi-agent Gathering environments (MAG) have complex state space than the other games, so that it is highly computation expensive to traverse the game tree to compute the \textsc{NashConv}. Instead, we evaluate all algorithms with a fixed policy set. In our experiments, the fixed policy set is generated with PSRO, i.e. $\Pi^{\text{PSRO}}$. We list the pseudo-code for evaluation in Algorithm~\ref{alg:evaluation}.

\begin{algorithm}[h]
 \caption{\textsc{Empirical Evaluation on MAG}}\label{alg:evaluation}
\KwInput{a policy set $\Pi^{\text{TEST}}$ of evaluated algorithm; $\Pi^{\text{PSRO}}$; an empty matrix $M \in \mathbb{R}^{|\Pi^{\text{TEST}}|\times|\Pi^{\text{PSRO}}|}$}

 	\For{each policy $\pi^{\text{TEST}}_i$ in $\Pi^{\text{TEST}}$}{
 		\For{each policy $\pi^{\text{PSRO}}_j$ in $\Pi^{\text{PSRO}}$}{
 		    Run 50 episodes to evaluate $M_{i,j}=u_i(\pi^{\text{TEST}}_i,\pi^{\text{PSRO}}_j)$
 		}
 		Compute score of $\sigma^{\text{TEST}}_{1:i}$ as $\textsc{Score}(\sigma^{\text{TEST}}_{1:i})=\sigma^{\text{TEST}}_{1:i}M_{1:i}\left[\sigma^{\text{PSRO}}\right]^T$
 	}
\KwOutput{a list of score $\textsc{Score}(\Pi^{\text{TEST}}) = \{\textsc{Score}(\sigma^{\text{TEST}}_{1:i}) \mid i = 1,\dots,|\Pi^{\text{TEST}}|\}$ for $\Pi^{\text{TEST}}$}
 \end{algorithm}

$\sigma^{\text{TEST}}_{1:i}$ in Algorithm~\ref{alg:evaluation} indicates a meta-strategy composed of $\pi_1,\dots,\pi_i$, and $M_{1:i}$ is a sub matrix with row 1 to $i$. We present $\textsc{Score}(\Pi^{\text{TEST}})$ of each algorithm in Figure~\ref{fig:gathering_score}. The curve of Self-Play is not included because we keep only two policies (one for the opponent, another for training) in our implementation.

\begin{figure}[h!]
	\centering
	\subfigure[Gathering Small]{\label{fig:gathering_traj_small}\includegraphics[width=.45\textwidth]{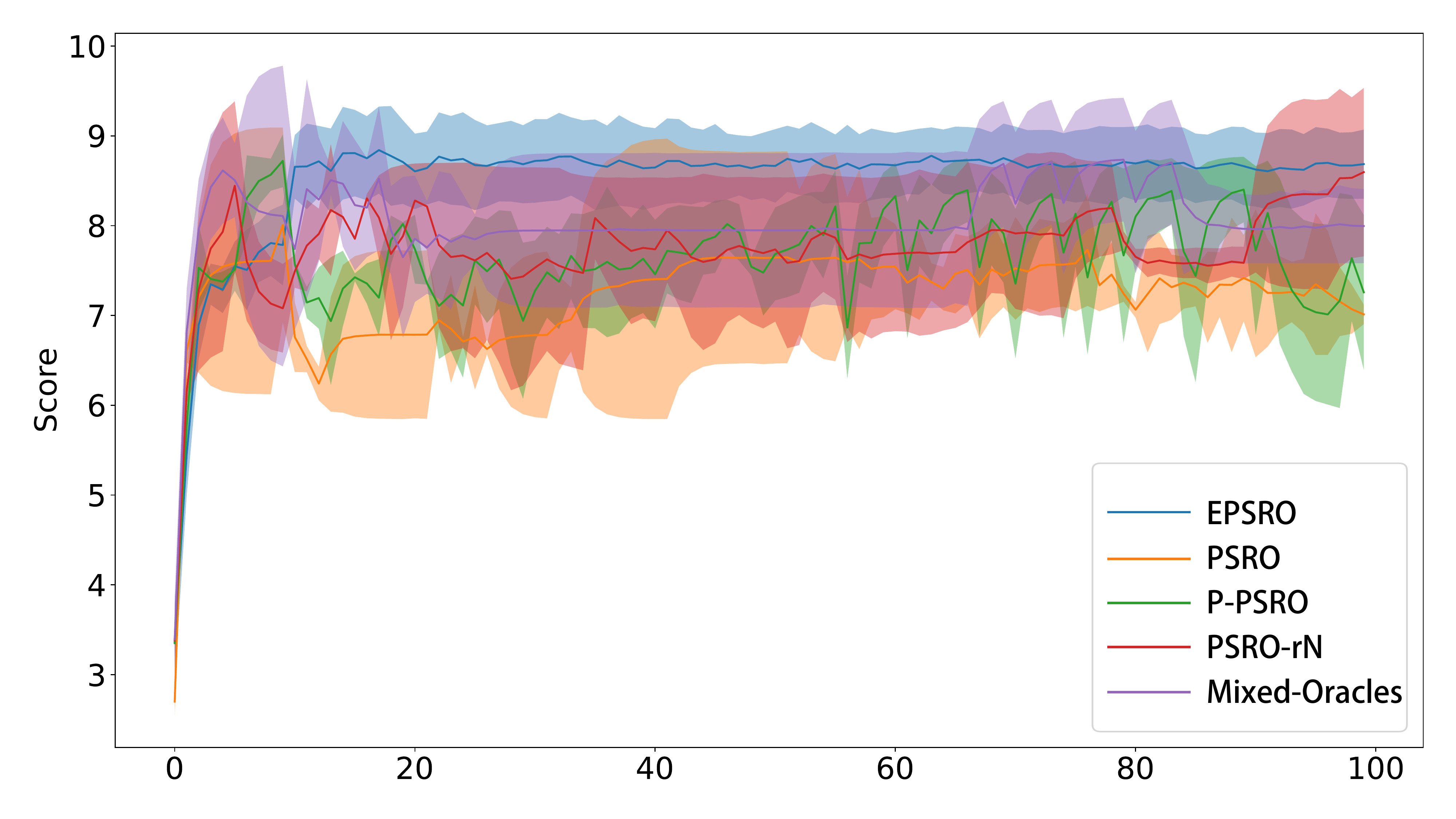}}
	\subfigure[Gathering Open]{\label{fig:gathering_traj_open}\includegraphics[width=.45\textwidth]{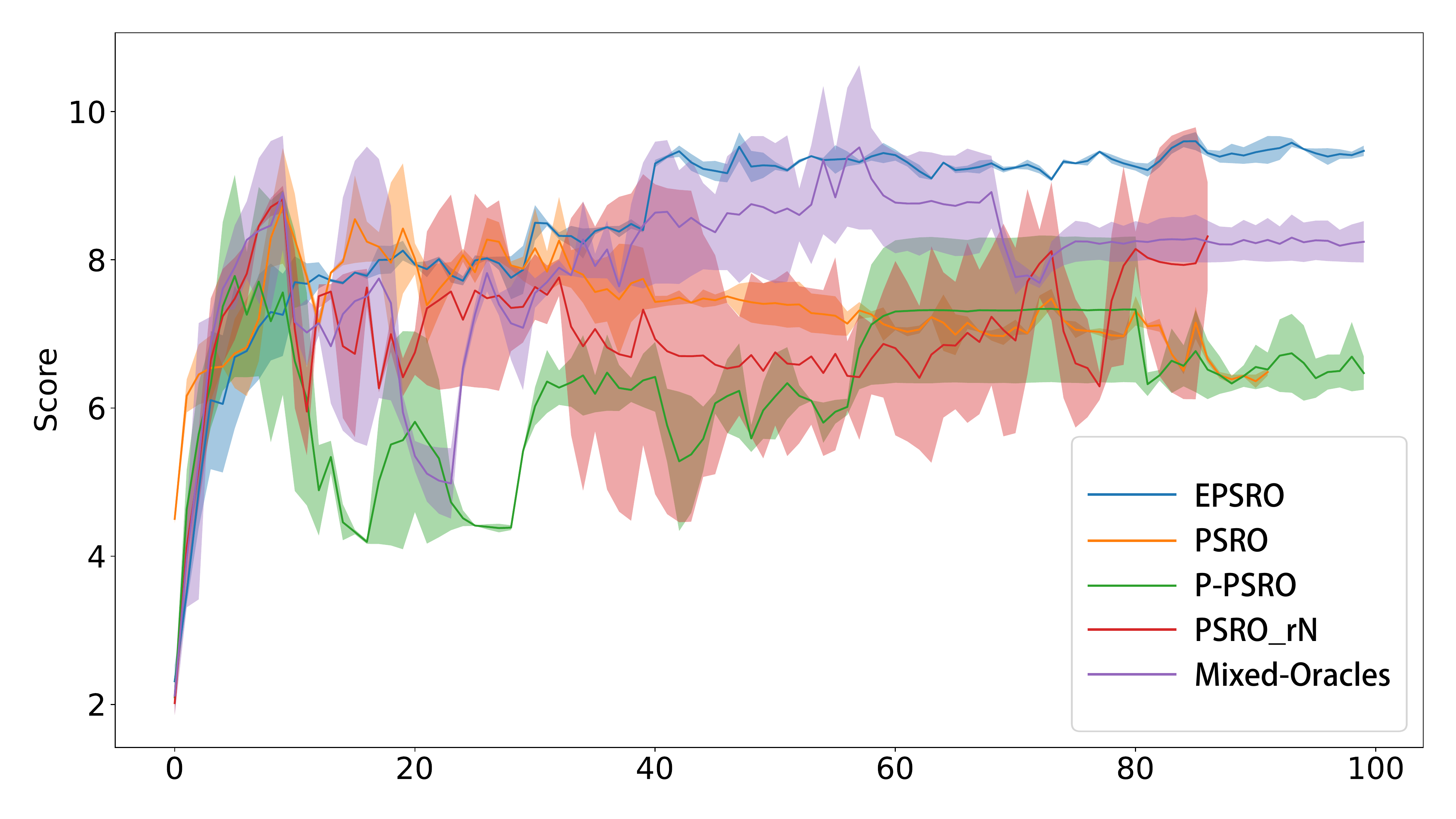}}
    \caption{The score of algorithms on two multi-agent gathering environments. The horizon axis indicates the number of training iterations. As reported in this Figure, \framework{} performs better than other algorithms.}
    \label{fig:gathering_traj}
\end{figure}

In the training stage, we set the number of simulations for each joint policy as 100, and 10000 episodes to optimize each policy. Except for \framework{}, the training time of each algorithm on the Gathering Small is about 24 hours, and 26 hours for the Gathering Open, while the training time for \framework{} is about 8 hours.

\section{Additional}

\subsection{Convergence Proof of Average Substitute Strategies}

\begin{theorem}{(\cite[Theorem~1]{brown2016strategy}).}\label{theorem:noam_theorem_2}
     In a two-player zero-sum game, if $\frac{R^T_i}{T} \le \epsilon_i$ for both player $i \in \{1,2\}$, then $\bar{\sigma}$ is a $(\epsilon_1 + \epsilon_2)$-equilibrium.
\end{theorem}
\begin{proof}
Follow the proof approach of Waugh et al. (2009). From (5), we have that
\begin{equation}
    \max_{\sigma\in\Delta}\frac{1}{T}\left( \sum^T_{t=1}u_i(\sigma'_i, \sigma^t_{-i}) - u_i(\sigma^t_i,\sigma^t_{-i}) \right) \le \epsilon_i.
\end{equation}
Since $\sigma'_i$ is the same on every iteration, this becomes
\begin{equation}\label{eq:noam_11}
    \max_{\sigma'\in\Delta}u_i(\sigma'_i,\bar{\sigma}^T_{-i}) - \frac{1}{T}\sum^T_{t=1}u_i(\sigma^t_i,\sigma^t_{-i}) \le \epsilon_i.
\end{equation}
Since $u_i(\sigma)=u_2(\sigma)$, if we sum Equation~\ref{eq:noam_11} for both players
\begin{align}
    \max_{\sigma'_1\in\Delta}u_1(\sigma'_1,\bar{\sigma}^T_2)+\max_{\sigma'_2\in\Delta}u_2(\bar{\sigma}^T_1,\sigma'_2) &\le \epsilon_1 + \epsilon_2,\label{eq:noam_2}\\
    \max_{\sigma'_1\in\Delta}u_1(\sigma'_1,\bar{\sigma}^T_2-\min_{\sigma'_2\in\Delta}u_1(\bar{\sigma}^T_1,\sigma'_2) &\le \epsilon_1 + \epsilon_2.
\end{align}
Since $u_1(\bar{\sigma}^T_1,\bar{\sigma}^T_2) \ge \min_{\sigma'_2 \in \Delta}u_1(\bar{\sigma}^T_1,\sigma'_2)$, so we have $\max_{\sigma'_1\in\Delta}u_1(\sigma'_1,\bar{\sigma}^T_2) - u_1(\bar{\sigma}^T_1,\bar{\sigma}^T_2) \le \epsilon_1 + \epsilon_2$. By symmetry, this is also true for player 2. Therefore, $\langle \bar{\sigma}^T_1,\bar{\sigma}^T_2 \rangle$ is a $(\epsilon_1 + \epsilon_2)$-equilibrium.
\end{proof}

When warm starting, we can calculate this value because we set $\bar{\sigma}^T=\sigma$. However we cannot calculate $\sum^T_{t=1}u_i(\sigma^t)$ because we did not define individual strategies played on each iteration. Fortunately, it turns out we can substitute another value we refer to as $Tu_i(\bar{\sigma})$, chosen from a range of acceptable options. To see this we first observe that the value of $\sum^T_{t=1}u_i(\sigma^t)$ is not relevant to the proof of Theorem~\ref{theorem:noam_theorem_2}. Specifically, in Eq.~\ref{eq:noam_2}, we see it cancels out. Thus, if we choose $(\bar{\pi}_o,\bar{\sigma}_{-i})$ such that satisfies it. Since $\max_{\pi}u_i(\pi, \bar{\sigma}^T_{-i}) \ge u_i(\bar{\pi}_i,\bar{\sigma}_{-i})$ and $\max_{\sigma_{-i}}u_{-i}(\bar{\pi}_i, \sigma_{-i}) \ge u_{-i}(\bar{\pi}_i,\bar{\sigma}_{-i})$. Thus $(\bar{\pi}_i,\bar{\sigma}_{-i})$ is a feasible warm-starting strategy.

\subsection{Comparison of Regret Bound}\label{sec:comparison_regret}

We compare \framework{} with existing solvers in the Table. Properties considered for comparison including (1) time complexity or regret bound; (2) whether we need to do retraining; (3) whether we need to know the full game.

\begin{center}
\begin{tabular}{cccc} \toprule
    \multirow{2}{*}{\textbf{Method}} & \multirow{2}{4cm}{\textbf{Time Complexity ($\tilde{\mathcal{O}}$) / Regret Bound ($\mathcal{O}$)}} & \multirow{2}{2cm}{\textbf{No Need to Re-train}} & \multirow{2}{3cm}{\textbf{No Need to Know the Full Game}} \\
    & & &\\ \toprule
    Linear Programming~\cite{van2020deterministic}  & $\tilde{\mathcal{O}}(n\exp{(-T/n^2.38)})$ & $\times$ & $\times$ \\ \midrule
    Fictitious Play~\cite{leslie2006generalised}  & $\tilde{\mathcal{O}}(T^{-1/(n+m-2)})$  & $\times$ & \checkmark  \\ \midrule
    Double Oracle~\cite{mcmahan2003planning}  & $\tilde{\mathcal{O}}(n\exp{(-T/n^3.88)})$ & $\times$ & \checkmark \\ \midrule
    Multipli. Weight Update\cite{freund1999adaptive}  & $\mathcal{O}(\sqrt{\log{n}/T})$  & $\times$ & \checkmark \\ \midrule
    Policy Response Oracles~\cite{lanctot2017unified}  & $\times$  & $\times$ & \checkmark \\ \midrule
    Online Double Oracle~\cite{le2021online} & $\mathcal{O}(\sqrt{k\log{k}/T})$ & $\times$ & \checkmark  \\ \midrule
    \textbf{\framework{}} & \pmb{$\mathcal{O}(\sqrt{\log{[(k^2+k)/2]}/T})$}  & \textbf{\checkmark} & \textbf{\checkmark} \\ \bottomrule
\end{tabular}
\end{center}


\end{document}